\documentclass[11pt,onecolumn,a4paper]{IEEEtran}

\usepackage{graphicx,epsfig,amsmath,amssymb,verbatim,color}
\usepackage{titlesec}
\usepackage{dsfont}
\usepackage{float}
\usepackage{tikz}
\usepackage{pgfplots}
\usepackage{xcolor}
\definecolor{darkred}{rgb}{0.8,0.1,0.1}
\usepackage{soul}
\usepackage{amsmath,mathtools}
\usepackage{enumerate}
\usepackage[T1]{fontenc}
\usepackage[utf8]{inputenc}
\usepackage{pifont}
\usepackage{hyperref}
\usepackage{tabularx}

\hypersetup{colorlinks=true,citecolor=blue,linkcolor=blue,filecolor=blue,urlcolor=blue,breaklinks=true}
\usepackage[utf8]{inputenc}  
\usepackage{fontenc}         
\usepackage{colortbl}
\usepackage{pifont}
\definecolor{Gray}{gray}{0.92}
\definecolor{Gray2}{gray}{0.75}
\definecolor{maroon}{cmyk}{0,0.87,0.68,0.32}
\usepackage{booktabs}
\usepackage{makecell}
\usepackage{diagbox}
\usepackage{multirow}

\newtheorem{definition}{Definition}
\newtheorem{proposition}{Proposition}
\newtheorem{lemma}[proposition]{Lemma}

\newtheorem{theorem}{Theorem}
\newtheorem{remark}{Remark}
\newtheorem{example}[proposition]{Example}
\newtheorem{corollary}{Corollary}

\newenvironment{proof}{\noindent \textit{{Proof.}~}}{\hfill $\square$}

\def\squareforqed{\hbox{\rlap{$\sqcap$}$\sqcup$}}
\def\qed{\ifmmode\squareforqed\else{\unskip\nobreak\hfil
		\penalty50\hskip1em\null\nobreak\hfil\squareforqed
		\parfillskip=0pt\finalhyphendemerits=0\endgraf}\fi}
\def\endenv{\ifmmode\;\else{\unskip\nobreak\hfil
		\penalty50\hskip1em\null\nobreak\hfil\;
		\parfillskip=0pt\finalhyphendemerits=0\endgraf}\fi}

\newcommand{\bra}[1]{\langle#1|}
\newcommand{\ket}[1]{|#1\rangle}

\def\Dbar{\leavevmode\lower.6ex\hbox to 0pt
	{\hskip-.23ex\accent"16\hss}D}

\begin{document}
	\title{Probabilistic Entanglement Distillation: Error Exponents via Postselected Quantum Hypothesis Testing Against Separable States}
	
	
	\author{Xian Shi
		\thanks{X. Shi is with College of Information Science and Technology,
			Beijing University of Chemical Technology, Beijing 100029, China.(email:\,shixian01@buct.edu.cn).
	}}
	
	%
	
	
	\maketitle
	\date{\today}
	\begin{abstract}
		Entanglement distillation is a fundamental task in quantum entanglement theory. While recent progress has clarified limitations of probabilistic transformations in general resource theories, an analytic formula for the error exponent of probabilistic entanglement distillation under approximately (dually) nonentangling operations has remained unavailable.  This work considers the error exponents of probabilistic entanglement distillation under the finite block length scenario and zero rate scenario when the operational model is  $\delta$-approximately nonentangling  or $\delta$-approximately dually nonentangling quantum instruments. 
		Building on the framework of postselected quantum hypothesis testing, we establish a direct connection between probabilistic distillation and postselected hypothesis testing against the set of separable states. In particular, we derive an analytical characterization of the distillation error exponent under $\delta$-(dually) approximately nonentangling quantum instruments.

	\end{abstract}

	\IEEEpeerreviewmaketitle

	\section{Introduction}
	Entanglement is a central resource in quantum information theory, underpinning fundamental tasks such as quantum key distribution, teleportation, and superdense coding \cite{horodecki2009quantum,plenio2014introduction,ekert1991quantum,bennett1992communication,bennett1993teleporting}.  A core objective is to quantify entanglement in an operationally meaningful way, and two canonical figures of merit are entanglement distillation and entanglement cost \cite{bennett1996mixed}. While these tasks are traditionally defined under local operations and classical communication (LOCC), the set of LOCC operations lacks a tractable characterization, which makes sharp information-theoretic analysis difficult \cite{Chitambar2014Everything}.  This has motivated the study of analytically convenient relaxations—such as PPT-preserving maps and (approximately) non-entangling operations—that retain operational relevance while enabling tight bounds and asymptotic characterizations \cite{peres,harrow2003,chitambar2020entanglement,brandao2010reversible,brandao2011one,regula2019one,fang2019non,chitambar2020entanglement,lami2023no,lami2024distillable,rippchen2025fundamental}. In particular, (approximately) non-entangling (NE) and (approximately) dually non-entangling (DNE) operations have recently emerged as a powerful framework for studying distillation and cost beyond LOCC \cite{lami2024distillable,shi2025,streltsov2025}.  Building on the information-theoretic paradigm that connects reliability limits to quantum hypothesis testing and error exponents \cite{bjelakovic2005quantum,notzel2014hypothesis,hayashi2025,hayashi2025general}, recent work further related the distillation error exponent under NE operations to composite entanglement testing exponents (Sanov-type) \cite{lami2024}. This development suggests that tools and intuitions from the resource detection theory into resource manipulations.

	A complementary line of research considers probabilistic (postselected) transformation protocols, where one allows a non-unit success probability in exchange for improved transformation accuracy. Probabilistic state transformations have been investigated in general resource theories \cite{regula2022probabilistic,regula2022tight,regula2023overcoming,lami2024rever}. In particular, Regula addressed general methods to characterize the transformations of quantums states with the aid of probabilistic protocols, there the author also presented the trade-off between the success probability and the errors of the transformations between two states \cite{regula2022probabilistic,regula2022tight}, and subsequent work provided exact asymptotic characterizations of probabilistic transformation limitations \cite{regula2023overcoming}. Nevertheless, to the best of our knowledge, an explicit analytic characterization of the error exponent for probabilistic entanglement distillation with zero rate under approximately non-entangling and approximately dually non-entangling operation models has remained open. The goal of this paper is to fill this gap by providing an information-theoretic exponent characterization for postselected entanglement distillation (and related cost quantities) under these approximately non-entangling frameworks.

	We consider probabilistic entanglement manipulation under instruments that have limited entangling power. Following the approximate nonentangling models, an instrument 
	$\mathcal{E}=\{\mathcal{E}_i\}$ is 
	$\delta$-approximately nonentangling if, for every separable input 
	$\sigma$, each subchannel output 
	$\mathcal{E}_i(\sigma)$ remains within entanglement budget 
	$\delta$ as quantified by the Hilbert projective metric to the separable set; additionally, 
	$\delta$-approximately dually nonentangling instruments impose a dual (cone-preserving) constraint. 
	
	For distillation, given blocklength 
	$n$ and target maximally entangled state 
	$\Psi_{m_n}$, we allow postselection on an instrument outcome 
	$i$, and require the postselected state 
	$\frac{\mathcal{E}_i(\rho^{\otimes n})}{\mathrm{ tr}\mathcal{E}_i(\rho^{\otimes n})}$ to achieve fidelity at least 
	$1-\epsilon_n$  with $\Psi_{m_n}.$ The performance metric is the optimal asymptotic error exponent 
	$-\frac{1}{n}\log\epsilon_n$, optimized over all feasible instruments and postselected outcomes, and further optimized in the limit 
	$n\rightarrow\infty$. 
	Our main results show that these operational error exponents admit exact information-theoretic characterizations in terms of reversed postselected hypothesis testing against the separable set. In particular, under 
	$\delta$-approximately nonentangling instruments, the distillation error exponent equals the regularization of reversed postselected-testing quantity defined via the Hilbert projective metric; under 
	$\delta$-approximately dually nonentangling instruments, the exponent equals the corresponding quantity with measurements restricted to be separable $(\mathbb{SEP}). $  
	
	We plot figure \ref{fig2} to summarize the conceptual bridge of this manuscript.

	\emph{Contributions:} This paper develops an information-theoretic framework for probabilistic entanglement manipulation under constraints that limit entangling power:
	\begin{itemize}
		\item[1.] (Exponent characterization under $\delta$-ANE instruments) We first present an analytical formula of the error exponent of entanglement distillation for the finite block length under $\delta$-approximately nonentangling instruments with the help of postselected quantum composite hypothesis testing.  We also prove that the optimal asymptotic error exponent of probabilistic entanglement distillation under $\delta$-approximately nonentangling instruments admits an exact characterization:
		\begin{align*}
			E^{\mathcal{NE}_{\delta}}_{d,err,p}(\rho_{AB})=&\hat{D}_{\Omega,Sep}^{reg}(\rho).
		\end{align*}
		Here $\hat{D}_{\Omega,Sep}^{reg}(\rho)$ is a regularized reversed postselected-testing quantity defined via $D_{\Omega}$(Hilbert projective metric-based). 
		\item[2.] (Separable-measurement counterpart under $\delta$-ADNE instruments) Under the $\delta$-approximately dually nonentangling instruments, based on an analogous exponent characterization, we present an analytical expression of the error exponent with the SEP-measurement-restricted version of the reversed postselected quantum composite testing quantity:
		
		\begin{align*}
			E_{d,err,p}^{\mathcal{DNE}_{\delta}}(\rho_{AB})=\hat{D}_{\Omega,Sep}^{reg,\mathbb{SEP}}(\rho).
		\end{align*}
		
		\item[3.] (A closed-form example) For Werner states, we provide an explicit evaluation of the exponent under the $\delta$-ANE instruments.
		
	\end{itemize}
	
	\begin{figure}[htbp] 
		\centering
		\includegraphics[width=0.7\textwidth]{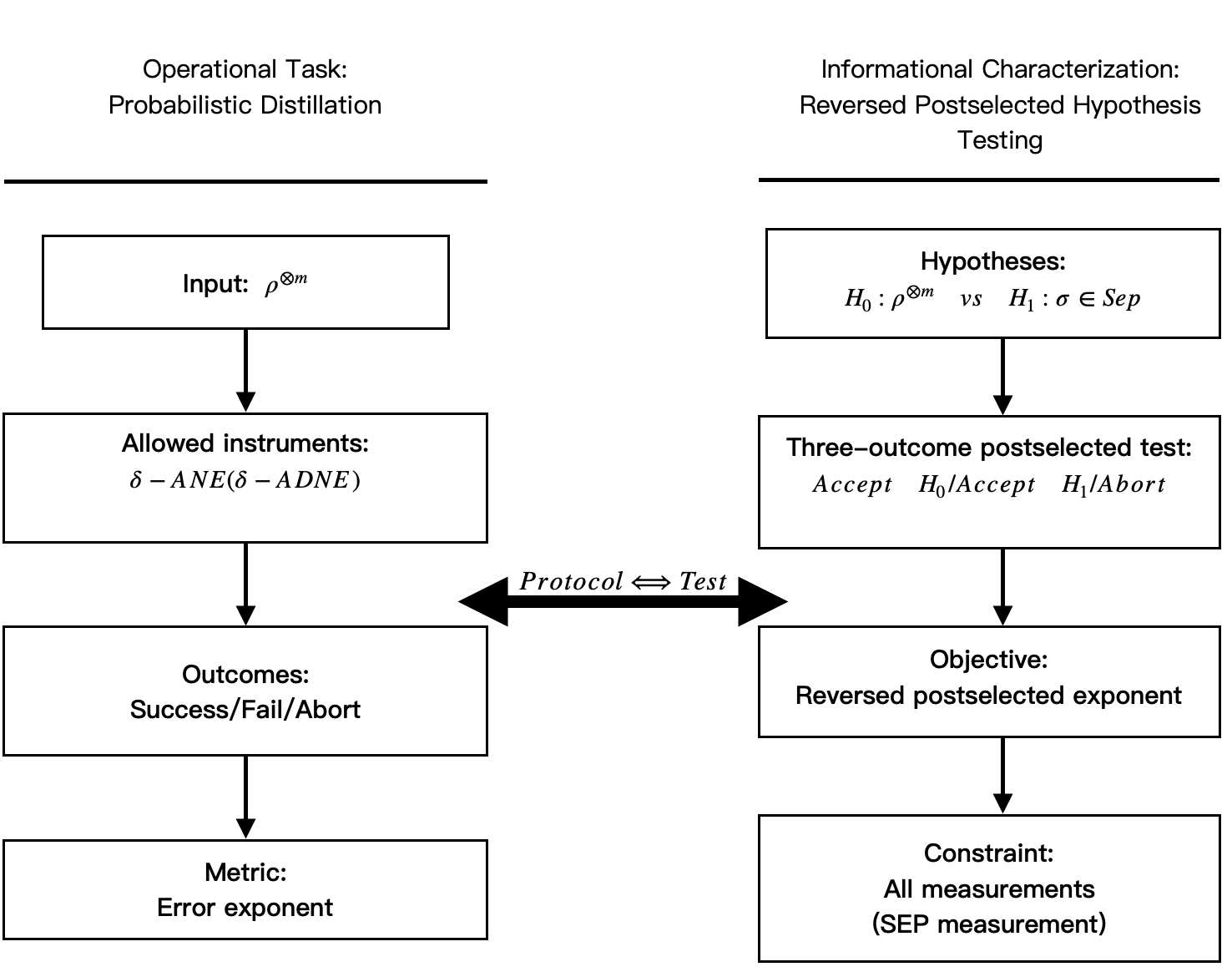} 
		\caption{Relations between the probabilistic entanglement distillation and postselected quantum hypothesis testing. }
		\label{fig2}
	\end{figure}
	
	\emph{Proof ideas:} \hspace{3mm} The bridge between probabilistic manipulation and postselected hypothesis testing \cite{regula2023postselected} is built via three-outcome postselected tests $M=(M_0,M_1,M_2)$, where outcome 1 and 2 correspond to deciding for two hypotheses and outcome 0 corresponds to abstention; performance is measured by conditional type-I/type-II errors conditioned on non-abstention. For the converse (meta-converse), any feasible probabilistic distillation protocol under 
	$\delta$-approximately (dually) nonentangling instruments induces a feasible postselected test against separable states, yielding an upper bound on the achievable exponent in terms of a reversed postselected-testing quantity. For achievability, we show how to construct a subchannel from a feasible postselected test so that the induced fidelity constraint matches the conditional-error constraint, leading to the matching lower bound and hence the exponent equalities stated in Theorems 1 and 2.  The 
	$\delta$-approximately dually nonentangling characterization additionally incorporates the dual constraint, which aligns the operational restriction with separable-measurement-restricted postselected testing. 
	
	\emph{Relation to prior work:} \hspace{3mm}Recent progress has clarified fundamental limitations of probabilistic transformations in general resource theories \cite{regula2022probabilistic,regula2022tight,regula2023overcoming}, yet an analytic exponent formula for probabilistic entanglement distillation under approximately (dually) non-entangling operations has remained unavailable. Our results close this gap by establishing an operational equivalence between probabilistic entanglement distillation and reversed postselected hypothesis testing against separable states. Comparing with \cite{lami2024distillable,lami2024}, we address the entanglement distillation under the probabilistic transformations. Compared with \cite{regula2023postselected}, our contribution is to identify a concrete operational task (probabilistic entanglement manipulation under limited entangling power) for which the optimal error exponent is exactly captured by a reversed postselected testing quantity (and, under $\delta$-ADNE, by its separable-measurement counterpart).
	
	\emph{Organization}. Section \ref{pk} introduces notation and reviews the operational models of (approximately) nonentangling and dually nonentangling instruments, as well as the probabilistic distillation exponents considered in this work. Section \ref{m} establishes the operational equivalence between probabilistic entanglement distillation under approximate (dually) nonentangling models and reversed postselected hypothesis testing, and presents the resulting exponent characterizations. Proofs of the main theorems and auxiliary properties of the Hilbert projective metric are deferred to the Appendix \ref{app}.

	\section{Preliminary Knowledge}\label{pk}
	
	Assume $\mathcal{H}_A$ is a Hilbert space with finite dimensions, which is relevant to the quantum system $A$. Let ${Herm}_{\mathcal{H}_A}$ and ${PSD}_{\mathcal{H}_A}$ be the set of Hermitian operators and positive semidefinite operators acting on $\mathcal{H}_A$, respectively. A quantum state is positive semidefinite with trace 1. Let $\mathcal{D}(\mathcal{H})$ be the set of quantum states acting on $\mathcal{H}$, $\mathcal{D}(\mathcal{H})=\{\rho|\rho\ge 0,\mathrm{ tr}\rho=1\}$. And $\overline{\mathcal{D}(\mathcal{H})}$ is the set of substates, which are semidefinite positive operators with trace less than 1, $\overline{\mathcal{D}(\mathcal{H})}=\{\gamma|\gamma\ge0,\mathrm{ tr}\gamma\le 1\}.$ Here $\gamma\ge \varphi$ denotes that $\gamma-\varphi$ is positive semidefinite. Assume $\vartheta$ is an operator of $\mathcal{H}_A$,{ let $ker(\vartheta)=\{\ket{\psi}|\vartheta|{\psi}\rangle=0\}$ be the kernel of $\vartheta,$ and $supp(\vartheta)=ker(\vartheta)^{\perp}$ denotes the support of $\vartheta$. } 
	
	A quantum channel $\Delta_{A\rightarrow B}$ is a completely positive and trace-preserving linear map from $\boldsymbol{D}_A$ to $\boldsymbol{D}_B$, and we denote {$\boldsymbol{CP}_{A\rightarrow B}$ as the set of completely positive and trace nonincreasing maps from $A$ to $B$ and $\boldsymbol{C}_{A\rightarrow B}$ as the set of quantum channels from $A$ to $B$.} In cases where no ambiguity arises, we generally denote $\boldsymbol{CP}$ and $\boldsymbol{C}$ as the set of all completely positive and trace-nonincreasing maps and channels, respectively.

	A positive operator valued measurement (POVM) $\{M_i|i=1,2,\cdots,k\}$ is a set of positive semidefinite operators with $k$ outcomes and $\sum_i M_i=\mathbb{I}.$ Here we denote $\mathcal{M}_k$ as the set of POVMs with $k$ outcomes. When a POVM applies to a state $\rho$, the probability to get the $r$-th outcome is given by 
	$P(r|\rho)=\mathrm{tr}M_r\rho.$ Moreover, each POVM $\{M_i\}_{i=1}^k$ can be regarded as a channel 
	\begin{align}
		\Lambda(\rho)=\sum_i\mathrm{tr}(M_i \rho)\ket{i}\bra{i},
	\end{align}
	which transforms a quantum state into a state acting on a classical system.

	Next we recall the Hilbert projective metric between two states $\rho$ and $\sigma$ $D_{\Omega}(\rho,\sigma)$ \cite{hilbert2011}, which is defined as follows,
	\begin{equation*}
		D_{\Omega}(\rho,\sigma)=\begin{cases}\log\inf\{ \gamma| \rho\le \eta\sigma\le \gamma\rho,\eta,\gamma\ge 0\}, & supp(\rho)=supp(\sigma)\\
			\infty,&\textit{otherwise}
		\end{cases}
	\end{equation*}
	Properties of $D_{\Omega}(\rho,\sigma)$ needed are presented in Lemma \ref{l1}.

	Assume $\mathcal{H}_{AB}$ is a bipartite system with finite dimensions. A state $\rho$ is separable if it can be written as $\rho=\sum_ip_i\rho_i^A\otimes\rho_i^B$, otherwise, it is entangled. Here we denote the set of separable states of $\mathcal{H}_{AB}$ as $Sep_{A:B}$. Besides, we denote $\overline{Sep}_{A:B}$ as the set of separable substates on $\mathcal{H}_{AB}$. When $dim(\mathcal{H}_A)=dim(\mathcal{H}_B)=d$, the maximally entangled state is $\ket{\psi}_d=\frac{1}{\sqrt{d}}\sum_{i=1}^d\ket{ii}$.

	Let $\mathcal{E}_i$ be completely positive and trace non-increasing, if  
	\begin{align*}
		\rho\in Sep_{A:B}\Longrightarrow	\frac{\mathcal{E}_i(\rho)}{\mathrm{tr}\mathcal{E}_i(\rho)}\in Sep_{A:B}, \forall i,
	\end{align*}
	then $\mathcal{E}_i$ is nonentangling($\mathcal{NE}$).  If $\{\mathcal{E}_i\}$ is a set of $\mathcal{\mathcal{NE}}$ subchannels and $\mathrm{tr}\sum_i\mathcal{E}_i(\cdot)=\mathrm{ tr}(\cdot),$ then $\{\mathcal{E}_i\}$ is a $\mathcal{NE}$ instrument. Here we denote the set of all such $\mathcal{NE}$ instruments as $\mathbb{O}_{\mathcal{NE}}.$ Next for a subchannel $\Lambda(\cdot)$, we can look at the Heisenberg picture, where $\Lambda^{\dagger}(\cdot)$ satisfies the following property, $\mathrm{ tr}X\Lambda(Y)=\mathrm{ tr}\Lambda^{\dagger}(X)Y$, for any $X$ and $Y.$ If $\Lambda$ satisfies the following property,
	\begin{align*}
		\Lambda(\rho)\in cone(Sep_{A:B}),\hspace{3mm}
		\Lambda^{\dagger}(\rho)\in cone(Sep_{A:B}),\hspace{3mm} \forall\rho\in Sep_{A:B},
	\end{align*}
	then we say $\Lambda$ is dually nonentangling($\mathcal{DNE}$). If $\{\Lambda_i\}$ is a set of subchannels with each $\Lambda_i$ $\mathcal{DNE}$ and $\mathrm{ tr}\sum_i\Lambda_i(\cdot)=\mathrm{ tr}(\cdot)$, $\{\Lambda_i\}$ is a $\mathcal{DNE}$ instruments. The set of all such $\mathcal{DNE}$ instruments are denoted as $\mathbb{O}_{\mathcal{DNE}}$. Following the work of Brandao and Plenio \cite{brandao2010reversible}, we present the definitions of $\delta$-approximately $\mathcal{NE}$ and $\mathcal{DNE}$ quantum instruments. 
	
	\begin{definition}
		Assume $\{\mathcal{E}_i\}$ is a set of subchannels such that $\sum_i\mathcal{E}_i$ is a channel. The $\delta$-approximately nonentangling quantum instrument, $\mathcal{E}=\{\mathcal{E}_i|\mathcal{E}_i\in \boldsymbol{CP}, \sum_i\mathcal{E}_i\in\boldsymbol{C}\}$, is defined as
		\begin{align*}
			\mathbb{O}_{\mathcal{NE}_{\delta}}=\{\mathcal{E}|D_{\Omega,\overline{Sep}}(\mathcal{E}_i(\sigma))\le \delta, \forall \mathcal{E}_i\in \mathcal{E},\forall\sigma\in Sep\},
		\end{align*}
		where $D_{\Omega,\overline{Sep}}(\sigma)$ is the Hilbert projective metric between $\sigma$ and the set of separable substates, whenever the support condition makes $D_{\Omega}(\cdot,\cdot)$ finite. Here $\delta$ represents the maximal entanglement that can be generated on separable inputs. Analogously, the the $\delta$-approximately dually nonentangling quantum instruments, $\mathcal{E}=\{\mathcal{E}_i|\mathcal{E}_i\in \boldsymbol{CP}, \sum_i\mathcal{E}_i\in\boldsymbol{C}\},$ is defined as
		\begin{align*}
			\mathbb{O}_{\mathcal{DNE}_{\delta}}=\{\mathcal{E}|\mathcal{E}_i^{\dagger}(Sep)\subset cone(Sep),\forall\mathcal{E}_i\in \mathcal{E}\}\cap\mathbb{O}_{\mathcal{NE}_{\delta}}.
		\end{align*}
		
	\end{definition}
	
 Here we address the following entanglement manipulation. Let $m\in \mathbb{N}$, the error exponent of the following entanglement manipulation is defined as follows,
 \begin{align*}
 E^{(m),\mathcal{F}_{\delta}}_{d,\mathrm{err},p}(\rho^{\otimes n})	=&\sup -\frac{1}{n}\log\epsilon_n\\
 	\textit{s. t.}	\hspace{6mm}&F\Big(\frac{\mathcal{E}^{(n)}_{i_n}(\rho^{\otimes n})}{tr[\mathcal{E}^{(n)}_{i_n}(\rho^{\otimes n})]},\Psi_{m}\Big)
 	\ge 1-\epsilon_n ,  \exists\ i_n,\\
 	&\mathcal{E}^{(n)}\in \mathbb{O}_{\mathcal{F}_{\delta}},\mathcal{F}=\{\mathcal{NE},\mathcal{DNE}\}.
 \end{align*}   The zero-rate probabilistic distillation error exponent under the $\mathcal{F}_{\delta}$ instruments $\{\mathcal{E}_i\}$ is defined as follows, for any $\{m_n\}_n$ with $\lim\limits_{n\rightarrow\infty}\frac{m_n}{n}=0,$
	\begin{align*}
		E^{\mathcal{F}_{\delta}}_{d,\mathrm{err},p}(\rho)
		:=& \sup \liminf_{n\to\infty} -\frac{1}{n}\log\epsilon_n\\
		\textit{s. t.}	\hspace{6mm}&F\Big(\frac{\mathcal{E}^{(n)}_{i_n}(\rho^{\otimes n})}{tr[\mathcal{E}^{(n)}_{i_n}(\rho^{\otimes n})]},\Psi_{m_n}\Big)
		\ge 1-\epsilon_n ,  \exists\ i_n,\\
		&\mathcal{E}^{(n)}\in \mathbb{O}_{\mathcal{F}_{\delta}},\mathcal{F}=\{\mathcal{NE},\mathcal{DNE}\}.
	\end{align*}
	where $\Psi_m=\ket{\psi}_m\bra{\psi}$, $F(\rho,\sigma)=||\sqrt{\rho}\sqrt{\sigma}||_1.$

	Throughout the manuscript, we denote $\log(\cdot)$ as $\log_2(\cdot)$.
	\section{Main results}\label{m}
	\textbf{Probabilistic entanglement distillation exponents}-- Now we will present our first main result, which provides an analytical formula for the probabilistic entanglement distillation under approximatley $\mathcal{NE}$ instruments. Here the dimension of target maximally entangled state is at least two.
	\begin{theorem}\label{t1}
			Assume $\rho_{AB}^{\otimes n}$ is a bipartite state on $\mathcal{H}_{AB}^{\otimes n}$. Let $m\in \mathbb{N}$, $\delta\ge 0$, and the distillation is limited to $\mathcal{NE}_{\delta}$. Then the error exponent of the entanglement distillation is 
		\begin{align*}
			E_{d,err,p}^{(m),\mathcal{NE}_{\delta}}(\rho^{\otimes n}_{AB})=\frac{1}{n}\hat{\beta}_{\frac{2^{\delta}}{2^{\delta}+m-1},Sep}(\rho_{AB}^{\otimes n}).
		\end{align*}

		Assume $\{m_n\}_n$ is a sequence of natural numbers with $\lim\limits_{n\rightarrow\infty}\frac{\log m_n}{n}=0,$ then
		\begin{align*}
			E_{d,err,p}^{\mathcal{NE}_{\delta}}(\rho_{AB})=\hat{D}_{\Omega,Sep}^{reg }(\rho_{AB}).
		\end{align*}
	\end{theorem}
	
	The proof of Theorem \ref{t1} is placed in the Appendix \ref{aa1}.
	\begin{remark}
		Without postselection and constrained nonentangling operations, the optimal exponent under $\mathcal{NE}$ constraints turns into Lemma 1 in \cite{lami2024}.
	\end{remark}
	
	\begin{example}\label{e1}
		Assume $\mathcal{H}_{AB}$ is a bipartite system with $dim(\mathcal{H}_A)=dim(\mathcal{H}_B)=d$, and $\rho_{AB}$ is the Werner state,
		\begin{align*}
			\rho_{p}=p\cdot\frac{2P_s}{d(d+1)}+(1-p)\cdot\frac{2P_{as}}{d(d-1)},
		\end{align*}
		here $P_s=\frac{I+F}{2}$, $P_{as}=\frac{I-F}{2}$, $F$ is the swap operator, $F=\sum_{ij}\ket{ij}\bra{ji}$. Then for each $n\in\mathbb{N}$,
		\begin{equation*}
			\hat{D}_{\Omega,Sep}^{reg}(\rho_p) = \begin{cases} 
				\log\frac{1-p}{p}\hspace{3mm}p<\frac{1}{2} \\
				0\hspace{12mm}p\ge\frac{1}{2}
			\end{cases}  .
		\end{equation*}
	\end{example}
	The proof of Example \ref{e1} is placed in the Appendix \ref{aa1}. 
	
	\begin{figure}[htbp] 
		\centering
		\includegraphics[width=0.5\textwidth]{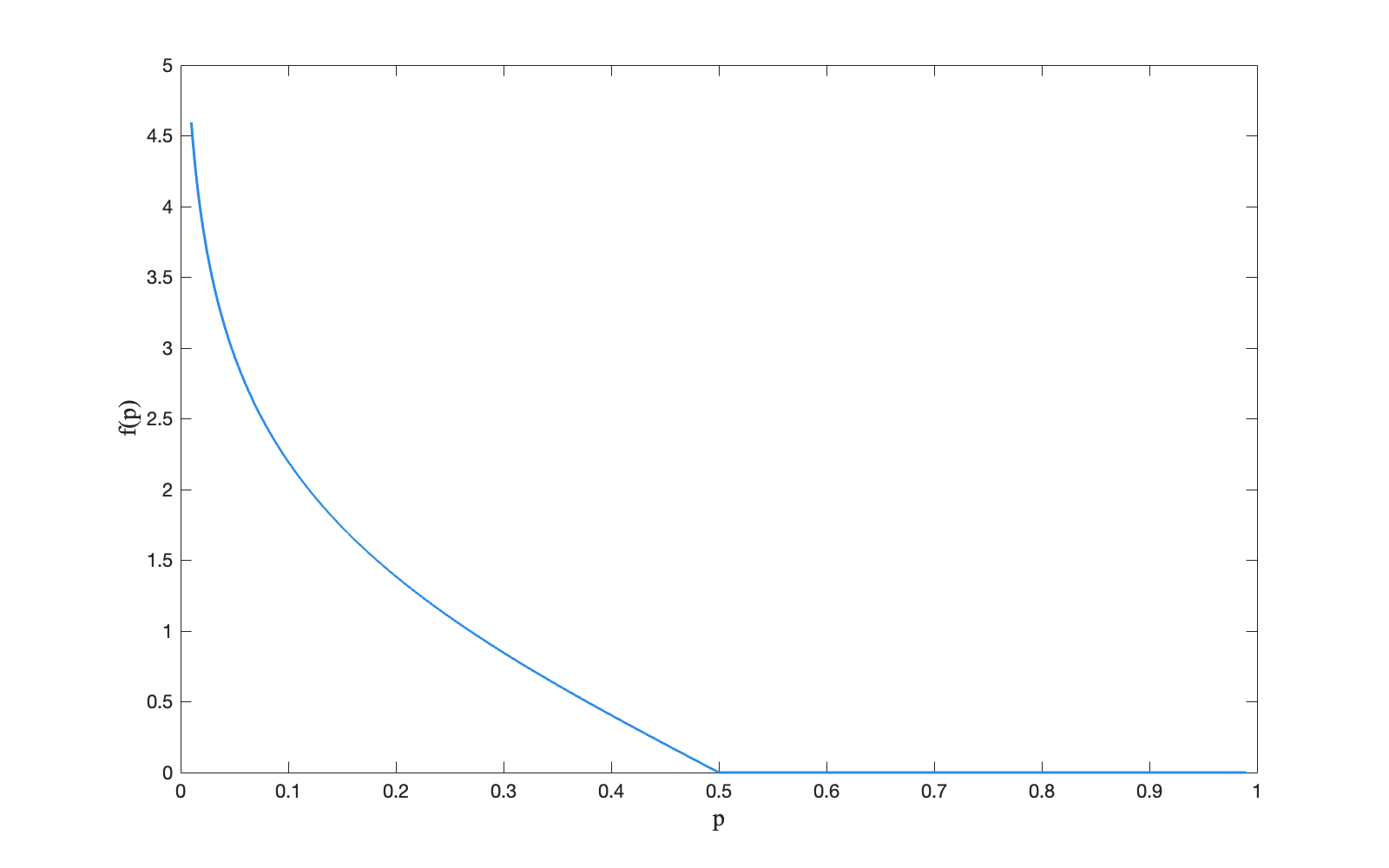} 
		\caption{$E_{d,err,p}^{\mathcal{NE}}(\rho_p)=\max\{0,\log\frac{1-p}{p}\}$ as a function of $p \in [0,1]$ for Werner state $\rho_p.$}
		\label{fig1}
	\end{figure}
	Based on Lemma \ref{l1} and Theorem \ref{t1}, we have
	\begin{align*}
		E_{d,err,p}(\rho_p)=	  \begin{cases} 
			\log\frac{1-p}{p}\hspace{3mm}p<\frac{1}{2} \\
			0\hspace{12mm}p\ge\frac{1}{2}
		\end{cases}  .
	\end{align*}

	Next when constraining the probabilistic entanglement distillation of a bipartite state under $\mathcal{DNE}_{\delta}$ instruments, we show the exponent is equal to the postselected quantum hypothesis testing for the state against the separable states when the measurements are restricted to be separable.
	\begin{theorem}\label{th2}
		Assume $\rho_{AB}^{\otimes n}$ is a bipartite state on $\mathcal{H}_{AB}^{\otimes n}$. Let $\delta\ge 0$ and $\{m_n\}_n$ is a sequence of natural numbers. Then the error exponent of the entanglement distillation is 
	\begin{align*}
		\hat{\beta}^{\mathbb{SEP}}_{\frac{1}{m_n+1}+\frac{2^{\delta}m_n}{(m_n+1)(2^{\delta}+m_n-1)},Sep}(\rho_{AB}^{\otimes n})- \log\frac{m_n+1}{m_n}\ge	nE_{d,err,p}^{(m_n),\mathcal{DNE}_{\delta}}(\rho_{AB}^{\otimes n})\ge& \hat{\beta}_{\frac{2^{\delta}}{2^{\delta}+m_n-1},Sep}^{\mathbb{SEP}}(\rho_{AB}^{\otimes n}).
	\end{align*}
	Furthermore, when $\{m_n\}_n$ is sequence with $\lim\limits_{n\rightarrow\infty}\frac{\log m_n}{n}=0,$ let $n\rightarrow \infty,$
	\begin{align*}
		E^{\mathcal{DNE}_{\delta}}_{d,err,p}(\rho_{AB})=\hat{D}_{\Omega,Sep}^{reg,\mathbb{SEP}}(\rho_{AB}).
	\end{align*}
	\end{theorem}
	
	The proof of Theorem \ref{th2} is placed in the Appendix \ref{aa1}.

	\section{Conclusion}
	In this work we investigated probabilistic entanglement manipulation on the asymptotic error exponents of probabilistic entanglement distillation under $\delta$-approximately nonentangling and $\delta$-approximately dually nonentangling  quantum instruments. Our main contribution is to present explicit analytical characterizations of the distillation error exponent by linking the operational task to postselected quantum hypothesis testing against the set of separable states. 
Collectively, our results provide a unified information-theoretic framework via postselected hypothesis testing of the probabilistic entanglement processing under the approximately nonentangling and approximately dually nonentangling instruments. We also present an analytical formula of the error exponent for the Werner states under $\delta$-approximately nonentangling quantum instruments.
	
	Several open directions remain. It would be interesting to (i) extend the present characterization to other resources, such as, coherence \cite{streltsov2017colloquium,lami2019,luo2025}, thermodynamics \cite{vinjanampathy2016quantum,deffner2019quantum}. (ii) develop efficiently computable semidefinite programming formulations \cite{wang2025computable,lami2025computable} for the relevant postselected testing quantities in practically regimes, and (iii) explore strong-converse \cite{cheng2024strong,oufkir2025quantum,berta2025channel} and second-order \cite{li2014second} refinements of the obtained exponents. We hope that the connection established here between probabilistic entanglement manipulation and postselected hypothesis testing will serve as a useful tool for further progress in operational entanglement theory.
	\section{Acknowledgement}
	X. S. was supported by the National Natural Science Foundation of China (Grant No. 12301580).
	\bibliographystyle{IEEEtran}
	\bibliography{ref}

	\clearpage

	\section{Appendix}\label{app}
	
	In this paper, $\mathcal{H}$ is a Hilbert space with finite dimensions. Let $Pos_{\mathcal{H}}$ be the set of positive semidefinite operators. When $\rho$ is positive semidefinite and $\mathrm{ tr}\rho=1,$  $\rho$ is a state on $\mathcal{H}$, when $\rho^{'}$ is positive semidefinite and $\mathrm{ tr}\rho\le 1$, $\rho^{'}$ is a substate. Here we denote $\mathcal{D}(\mathcal{H})$ and $\overline{\mathcal{D}(\mathcal{H})}$ as the set of states and substates, respectively. Assume $\mathcal{N}:Pos_{\mathcal{H}}\rightarrow Pos_{\mathcal{H}}$ is a map, $\mathcal{N}$ is positive. For any $d\in \mathbb{N}$, $\mathcal{I}_d\otimes\mathcal{N}$ is positive, we say $\mathcal{N}$ is completely positive. If $\mathrm{ tr}\mathcal{N}(\cdot)=\mathrm{tr}(\cdot)$, $\mathcal{N}$ is trace-preserving. And if $\mathrm{ tr}\mathcal{N}(\cdot)\le \mathrm{ tr}(\cdot)$, $\mathcal{N}$ is trace-nonpreserving. When $\mathcal{N}$ is completely positive and trace-preserving, we say $\mathcal{N}$ is a channel.

	\subsection{Entanglement}
	
	Assume $\mathcal{H}_{AB}$ is a bipartite system. A state $\rho_{AB}$ is separable if it can be written as $$\rho=\sum_i p_i\rho_i^A\otimes\rho_i^B,$$ 
	here the states $\rho^A_i$ and $\rho_i^B$ are states on local systems $A$ and $B,$ respectively. Otherwise, $\rho_{AB}$ is entangled. We will denote the set of separable states of $\mathcal{H}_{AB}$ as $Sep_{A:B}$, or simply $Sep$ if there is no ambiguity regarding the system.otherwise, it is entangled. Besides, we denote $\overline{Sep}$ and $cone(Sep)=\{\lambda\sigma|\lambda>0,\sigma\in Sep\}$ as the set of separable substates and cone of separable states.
	
	An important method to detect whether a state is separable is the positive partial transpose(PPT) criterion \cite{peres}, which said any separable state $\rho_{AB}$ satisfies the following inequality $\rho_{AB}^{T_B}$. A bipartite state $\sigma$ satisfying the PPT criterion is called a PPT state. Furthermore, we can generalize the above concepts to the POVMs. A measurement $\mathsf{M}$ is said to be separable measurements if 
	\begin{align*}
		\mathsf{M}=\{M_x|\sum_x M_x=\mathbb{I}, M_x\in \overline{Sep}\}.
	\end{align*}
	Here we denote the set of all separable measurements as $\mathbb{SEP}.$ Besides, we denote $\mathbb{ALL}$ as the set of all measurements, $\mathbb{ALL}=\{\{M_x\}_x|\sum_xM_x=\mathbb{I},M_x\ge 0,\forall x\}$.
	
	\subsection{Quantum Relative Entropies}
	Assume $\rho$ and $\sigma$ are two states, let $\alpha\in (1,\infty]$, then the $\alpha$-sandwiched Renyi divergence $\tilde{D}_{\alpha}(\rho,\sigma)$ for $\rho$ and $\sigma$ is defined as 
	\begin{equation*}
		\tilde{D}_{\alpha}(\rho,\sigma)=
		\begin{cases}
			\frac{\alpha}{\alpha-1}\log||\sigma^{\frac{1-\alpha}{2\alpha}\rho\sigma^{\frac{1-\alpha}{2\alpha}}}||_{\alpha}\hspace{3mm} \textit{if $supp(\rho)\subseteq supp(\sigma)$,}\\
			+\infty \hspace{5mm} \textit{otherwise},
		\end{cases}
	\end{equation*}
	when $\alpha\rightarrow1$, $\tilde{D}_{\alpha}(\rho,\sigma)$ tends to the quantum relative entropy of $\rho$ and $\sigma$, ${D}(\rho||\sigma)=\mathrm{tr}[\rho(\log\rho-\log\sigma)].$ 
	
	Next we define the other quantum relative entropy for two states $\rho$ and $\sigma$ with $supp(\rho)\subseteq supp(\sigma)$, $D_{max}(\rho,\sigma),$
	
	\begin{align}
		D_{max}(\rho,\sigma)=\log\inf& \hspace{2mm}\lambda\label{dpmax}\\
		\textit{s. t.}\hspace{5mm}&\rho\le\lambda\sigma\nonumber\\
		&\lambda\in \mathbb{R}^{+},\nonumber
	\end{align}
	otherwise, $D_{max}(\rho,\sigma)$ tends to the infty. 
	The dual program of $(\ref{dpmax})$ is 
	\begin{align}
		D_{max}(\rho,\sigma)=\log\max&\hspace{3mm}\mathrm{ tr}\rho X\label{ddmax}\\
		\textit{s. t.}\hspace{5mm}&\mathrm{ tr}\sigma X\le 1,\nonumber\\
		&X\ge0.\nonumber
	\end{align}
	
	The Hilbert projective metric between two states $\rho$ and $\sigma$ is
	\begin{align*}
		D_{\Omega}(\rho,\sigma)=D_{max}(\rho,\sigma)+D_{max}(\sigma,\rho).
	\end{align*}Let $\Omega(\rho,\sigma)=2^{D_{\Omega}(\rho,\sigma)}$. Assume $\mathcal{F}$ is a set of convex and compact set of states on $\mathcal{H}$, and $\rho\in \mathcal{D}(\mathcal{H})$, then
	\begin{align*}
		D_{\Omega,\mathcal{F}}(\rho)=\min_{\sigma\in\mathcal{F}} D_{\Omega}(\rho,\sigma),
	\end{align*} 
	where the minimum takes over all the states in $\mathcal{F}$.

	Then we introduce some quantities necessary for the results we obtained below. Assume $\rho$ and $\sigma$ are two states acting on $\mathcal{H},$ trace norm and fidelity are common used tools to show the distances between $\rho$ and $\sigma$. The trace norm distance between $\rho$ and $\sigma$ is defined as
	\begin{align*}
		||\rho-\sigma||_1=&\mathrm{ tr}\sqrt{(\rho-\sigma)^{\dagger}(\rho-\sigma)}\\
		=&\max_{||B||_{\infty}\le 1}|\mathrm{ tr}B(\rho-\sigma)|.
	\end{align*}
	The fidelity between quantum states $\rho$ and $\sigma$ are defined as
	\begin{align*}
		F(\rho,\sigma)=\mathrm{ tr}\sqrt{\sqrt{\rho}\sigma\sqrt{\rho}}=||\sqrt{\rho}\sqrt{\sigma}||_1,
	\end{align*}
	when $\rho=\ket{\psi}\bra{\psi}$ is a pure state, $F(\ket{\psi},\sigma)=\bra{\psi}\sigma\ket{\psi}.$
	
	Next we present the following properties of $D_{\Omega}(\rho,\sigma)$.
	\begin{lemma}\label{l1}
		Assume $\rho$ and $\sigma$ are two states, then
		\begin{itemize}
			\item[(1.)] $D_{\Omega}(\rho,\sigma)\ge 0$, and the quality happens if and only if $\rho=\sigma$.
			\item[(2.)] $D_{\Omega}(\rho,\sigma)=D_{\Omega}(\sigma,\rho)$.
			\item[(3.)] For arbitrary positive numbers $\lambda$ and $\varphi$, then $D_{\Omega}(\rho,\sigma)=D_{\Omega}(\lambda\rho,\varphi\sigma).$
			\item[(4.)] The quantity $D_{\Omega}(\cdot,\cdot)$ satisifies the data-processing property under the positive map, that is, for each positive linear map $\mathcal{E}$, 
			\begin{align*}
				D_{\Omega}(\mathcal{E}(\rho),\mathcal{E}(\sigma))\le D_{\Omega}(\rho,\sigma).
			\end{align*}
			\item[(5.)] $D_{\Omega}(\rho,\sigma)$ can be computed under the semidefinite programming method, 
			\begin{align}
				D_{\Omega}(\rho,\sigma)=&\log\sup \mathrm{ tr}A\rho\label{f1}\\
				\textit{s. t.} \hspace{4mm}& \mathrm{ tr}B\rho=1,\nonumber\\
				&\mathrm{ tr}(B-A)\sigma\ge0,\nonumber\\
				&A,B\ge 0\nonumber
			\end{align}
			\item[(6.)] Assume $\rho$ and $\sigma$ are two states, $D_{\Omega}(\rho^{\otimes n},\sigma^{\otimes n})=nD_{\Omega}(\rho,\sigma)$.
			\item[(7.)] $D_{\Omega}(\rho,\sigma)<\infty$ is valid if and only if \emph{supp}$(\rho)=$\emph{supp}$(\sigma)$.  
		\end{itemize}
	\end{lemma}

	\subsection{ Postselected Quantum Hypothesis Testing}
	Quantum state discrimination is a fundamental quantum information task. Recently, the authors in \cite{regula2023postselected} addressed the following problem. Assume Alice receives a state, and she knows that the state is $\rho$ or $\sigma$, her aim is to determine which state she obtained. In the scenario, she can perform a three-outcome positive operator-valued measure (POVM), $M=\{M_1,M_2,M_0\}.$ The outcome 1 and 2 correspond to the state $\rho$ and $\sigma$, respectively, when the outcome is 0, we cannot make a decision. Then they defined the following quantities,
	\begin{align*}
		\textit{conditional type I error:}\hspace{5mm}\overline{\alpha}(M)=\frac{\mathrm{ tr}M_2\rho}{\mathrm{ tr}(M_1+M_2)\rho},\\
		\textit{conditional type II error:}\hspace{5mm}\overline{\beta}(M)=\frac{\mathrm{ tr}M_1\sigma}{\mathrm{ tr}(M_1+M_2)\sigma},
	\end{align*}
	Assume $\mathcal{F}$ is a convex and closed set of quantum states, and the postselected hypothesis testing between a state $\rho$ and the set $\mathcal{F}$ is 
	
	\begin{align}
		\overline{\beta}_{\epsilon,\mathcal{F}}(\rho)=-\log\inf_{M\in \mathcal{M}_3}\{\sup_{\sigma\in\mathcal{F}}\frac{\mathrm{ tr}M_1\sigma}{\mathrm{ tr}(M_1+M_2)\sigma}|\frac{\mathrm{ tr}M_2\rho}{\mathrm{ tr}(M_1+M_2)\rho}\le\epsilon\}\label{pht}
	\end{align}
	
	where $M$ takes over all the elements in $\mathcal{M}_3$, and $\mathrm{ tr}(M_1+M_2)\sigma,\mathrm{ tr}(M_1+M_2)\rho> 0.$

	Here we address a reversed problem of the composite postselected hypothesis testing. Assume $\mathcal{F}$ is a convex and compact set, $M\in\mathcal{M}_3$ is a feasible POVM, the conditional type $II$ error is defined as
	\begin{align*}
		\overline{\beta}(M)=\frac{\mathrm{tr}M_1\rho}{\mathrm{ tr}(M_1+M_2)\rho},
	\end{align*}
	while the conditional type $I$ error is defined as $$\overline{\alpha}_{\mathcal{F}}(M)=\sup_{\sigma\in\mathcal{F}}\frac{\mathrm{tr}M_2\sigma}{\mathrm{tr}(M_1+M_2)\sigma}.$$
	The reversed composite postselected hypothesis testing, $\hat{\beta}_{\epsilon,\mathcal{F}}(\rho)$, is defined as follows,
	\begin{align}
		\hat{\beta}_{\epsilon,\mathcal{F}}(\rho)=&-\log\inf_{M\in\mathcal{M}_3}\frac{\mathrm{tr}M_1\rho}{\mathrm{ tr}(M_1+M_2)\rho}\label{lf0}\\
		\textit{s. t.}&\hspace{4mm} \frac{\mathrm{tr}M_2\sigma}{\mathrm{tr}(M_1+M_2)\sigma}\le\epsilon, \forall\sigma\in\mathcal{F},\nonumber\\
		&\hspace{4mm}0\le M_1+M_2\le \mathbb{I}.\nonumber
	\end{align}
	
	Furthermore, when $\mathcal{M}_3$ in (\ref{lf0}) is in a class of $\mathbb{M}$, then we define $\hat{\beta}_{\epsilon,\mathcal{F}}^{\mathbb{M}}(\rho)$ as follows
	\begin{align}
		\hat{\beta}^{\mathbb{M}}_{\epsilon,\mathcal{F}}(\rho)=&-\log\inf_{M\in\mathcal{M}_3}\frac{\mathrm{tr}M_1\rho}{\mathrm{ tr}(M_1+M_2)\rho}\label{lfm}\\
		\textit{s. t.}&\hspace{4mm} \frac{\mathrm{tr}M_2\sigma}{\mathrm{tr}(M_1+M_2)\sigma}\le\epsilon, \forall\sigma\in\mathcal{F},\nonumber\\
		&\hspace{4mm}0\le M_1+M_2\le \mathbb{I}, M\in \mathbb{M}.\nonumber
	\end{align}
	
	The analytical formula of $\hat{\beta}_{\epsilon,\mathcal{F}}(\rho)$ is presented in the following corollary.
	
	\begin{corollary}\label{c1}
		Assume $\mathcal{F}$ is a convex and compact set of quantum states, then 
		\begin{align}
			{\hat{\beta}_{\epsilon,\mathcal{F}}(\rho)}
			=\log[1+\frac{\epsilon}{1-\epsilon}{\hat{\Omega}}_{\mathcal{F}}(\rho)]
		\end{align}
		
		When each family set $(\mathcal{F}_n)_n$ are convex and compact, and $(\mathcal{F}_n)_n$ is closed under tensor product, let $\{\epsilon_n\}_n$ be a sequence of positive numbers with $\lim\limits_{n\rightarrow\infty}\frac{1}{n}\log{\epsilon_n}=\lim\limits_{n\rightarrow\infty}\frac{1}{n}\log\frac{1}{1-\epsilon_n}=0,$ we have
		\begin{align}
			\lim\limits_{n\rightarrow \infty}\frac{1}{n}{\hat{\beta}_{\epsilon_n,\mathcal{F}}(\rho^{\otimes n})}=\hat{D}_{\Omega,\mathcal{F}}^{reg}(\rho):=
			\lim\limits_{n\rightarrow\infty}\frac{1}{n}\log\hat{\Omega}_{\mathcal{F}}(\rho^{\otimes n}).
		\end{align}
	\end{corollary}
	\begin{proof}
		Here we take a similar method in \cite{regula2023postselected} to show the theorem. Based on the definition of $\hat{\beta}_{\epsilon,\mathcal{F}}(\rho),$ we have
		
		\begin{align*}
			&\hat{\beta}_{\epsilon,\mathcal{F}}(\rho)\\
			=&-\log\inf_{M\in \mathcal{M}_3}\{\frac{\mathrm{tr}M_1\rho}{\mathrm{tr}(M_1+M_2)\rho}|\frac{\mathrm{tr}M_2\sigma}{\mathrm{tr}(M_1+M_2)\sigma}\le\epsilon, \forall\sigma\in\mathcal{F}, 0\le M_1+M_2\le \mathbb{I}\}\\
			=&-\log\inf_{t,M\in \mathcal{M}_3}\{t|\frac{\mathrm{tr}M_1\rho}{\mathrm{tr}(M_1+M_2)\rho}\le t,\frac{\mathrm{tr}M_2\sigma}{\mathrm{tr}(M_1+M_2)\sigma}\le\epsilon, \forall\sigma\in\mathcal{F}, 0\le M_1+M_2\le \mathbb{I}\}\\
			=&-\log\inf_{\tilde{t},M_1,M_2^{'}\ge 0}\{\frac{1}{\tilde{t}}|\frac{\mathrm{tr}M_2^{'}\rho}{\mathrm{tr}M_1\rho}\ge 1,\frac{\mathrm{ tr}M_1\sigma}{\mathrm{ tr}M_2^{'}\sigma}\ge \frac{1-\epsilon}{\epsilon}(t^{'}-1),\forall\sigma\in \mathcal{F}\},
		\end{align*}
		
		In the third equality, we denote $M_2^{'}=\frac{t }{1-t}M_2$, in the last equality, $\tilde{t}=\frac{1}{t}$. Then we have
		\begin{align*}
			&2^{\hat{\beta}_{\epsilon,\mathcal{F}}(\rho)}\\
			=&\inf_{\sigma\in\mathcal{F}}\sup_{\tilde{t}\ge 0,M_1,M_2^{'}\ge 0}\{\tilde{t}|\frac{\mathrm{ tr}M_1\rho}{\mathrm{ tr}M_2^{'}\rho}\le 1,\frac{\mathrm{ tr}M_1\sigma}{\mathrm{ tr}M_2^{'}\sigma}\ge \frac{1-\epsilon}{\epsilon}(t^{'}-1)\}\\
			=&\inf_{\sigma\in \mathcal{F}}\sup_{\tilde{t}\ge 0,M_1,M_2^{'}\ge 0}\{   1+\frac{\epsilon}{1-\epsilon}\frac{\mathrm{ tr}M_1\sigma}{\mathrm{ tr}M_2^{'}\sigma}|\frac{\mathrm{ tr}M_1\rho}{\mathrm{ tr}M_2^{'}\rho}\le 1\}.
		\end{align*}
		As 
		\begin{align*}
			\min_{\rho\in \mathcal{F}}\Omega(\rho,\sigma)=&\sup_{A,B}\{\frac{\mathrm{tr}A\rho}{\mathrm{tr}B\rho}|\frac{\mathrm{tr}A\sigma}{\mathrm{ tr}B\sigma}\le 1,\forall\rho\in \mathcal{F}\}\\
			=&\inf_{\rho\in\mathcal{F}}\sup_{A,B}\{\frac{\mathrm{tr}A\rho}{\mathrm{tr}B\rho}|\frac{\mathrm{tr}A\sigma}{\mathrm{ tr}B\sigma}\le 1\},
		\end{align*}
		then
		\begin{align*}
			2^{\hat{\beta}_{\epsilon,\mathcal{F}}(\rho)}
			=&1+\frac{\epsilon}{1-\epsilon}\min_{\sigma\in\mathcal{F}}\Omega(\sigma,\rho)
		\end{align*}
		
		When for a generic $n$, let $\sigma_n\in\mathcal{F}_n$ be the optimal for $\rho^{\otimes n}$ in terms of $\Omega(\cdot,\rho^{\otimes n})$, 
		\begin{align}
			&	\frac{\epsilon}{1-\epsilon}\Omega(\sigma_n,\rho^{\otimes n})\le 2^{\hat{\beta}_{\epsilon,\mathcal{F}}(\rho^{\otimes n})}\le \frac{1}{1-\epsilon}\Omega(\sigma_n,\rho^{\otimes n})\nonumber\\
			\Longrightarrow&
			\log\frac{\epsilon}{1-\epsilon}+\log\Omega(\sigma_n,\rho^{\otimes n})\le {\hat{\beta}_{\epsilon,\mathcal{F}}(\rho^{\otimes n})}\nonumber\\\le& \log\frac{1}{1-\epsilon}+\log\Omega(\sigma_n,\rho^{\otimes n}),\label{lf00}
		\end{align}
		next when $\mathcal{F}_n$ is closed under tensor product, $\hat{\Omega}_{\mathcal{F}_{m+n}}(\rho^{\otimes m+n})\le \hat{\Omega}_{\mathcal{F}_{m}}(\rho^{\otimes m})+\hat{\Omega}_{\mathcal{F}_{n}}(\rho^{\otimes n})$, due to Fekete's Lemma, $\lim\limits_{n\rightarrow\infty}\hat{\Omega}_{\mathcal{F}_{n}}(\rho^{\otimes n})$ exists. Then dividing n to both sides of (\ref{lf00}) and taking the limit, we have
		\begin{align*}
			\lim\limits_{n\rightarrow\infty}\frac{1}{n}\hat{\beta}_{\epsilon_n,\mathcal{F}}(\rho^{\otimes n})=\lim\limits_{n\rightarrow\infty}\frac{1}{n}\log\hat{\Omega}_{\mathcal{F}}(\rho^{\otimes n}).
		\end{align*}
	\end{proof}

	\begin{corollary}\label{c2}
		Assume $\mathcal{F}$ is a convex and compact set of quantum states on $\mathcal{H}$ with $\frac{I}{d}\in \mathcal{F}$, here $d$ is the dimension of $\mathcal{H},$ $\mathbb{M}$ is a class of measurements on $\mathcal{H}$ with $(\mathbb{M}^{*})^{*}=cone(\mathbb{M})$, then 
		\begin{align}
			{\hat{\beta}^{\mathbb{M}}_{\epsilon,\mathcal{F}}(\rho)}
			=\log[1+\frac{\epsilon}{1-\epsilon}{\hat{\Omega}}^{\mathbb{M}}_{\mathcal{F}}(\rho)],
		\end{align}
		here 
		\begin{align*}
				\hat{\Omega}^{\mathbb{M}}_{\mathcal{F}}(\rho)=&\inf   \gamma\\
				\textit{s. t.}\hspace{3mm}&\rho\preceq_{\mathbb{M}^{*}}\tilde{\sigma}\preceq_{\mathbb{M}^{*}}\gamma\rho,\\
			&	\tilde{\sigma}\in cone(\mathcal{F}),
		\end{align*}
	here $X\preceq_{\mathbb{M}^{*}}Y$ means that $\mathrm{ tr}(Y-X)H\ge 0,$ $\forall H\in M$ and $M\in \mathbb{M}$.

		When the family set $(\mathbb{M}_n)_n$ and $(\mathcal{F}_n)_n$ are convex and compact, $(\mathbb{M}_n)_n$ and $(\mathcal{F}_n)_n$ are closed under tensor product, we have
		\begin{align}
			\lim\limits_{n\rightarrow \infty}\frac{1}{n}{\hat{\beta}^{\mathbb{M}}_{\epsilon,\mathcal{F}}(\rho^{\otimes n})}=\hat{D}_{\Omega,\mathcal{F}}^{reg,\mathbb{M}}(\rho):=
			\lim\limits_{n\rightarrow\infty}\frac{1}{n}\log\hat{\Omega}^{\mathbb{M}}_{\mathcal{F}}(\rho^{\otimes n}).
		\end{align}
		
	\end{corollary}

	First we present the dual form of $\Omega_{\mathcal{F}}^{\mathbb{M}}(\rho)$, which can be obtained through standard Lagrange duality arguments \cite{ponstein2004approaches},
	
		\begin{align*}
		\hat{\Omega}^{\mathbb{M}}_{\mathcal{F}}(\rho)=&\sup \mathrm{ tr}A\rho\\
		\textit{s. t.} \hspace{4mm}& \mathrm{ tr}B\rho=1,\nonumber\\
		&\mathrm{ tr}(B-A)\sigma\ge0,\nonumber\\
		&A,B\in\{\omega|\omega=\sum_i\mu_i M_i, \mu_i \ge 0\},\\
		&M=\{M_i\}_i\in  \mathbb{M}.\nonumber
	\end{align*}

The strong duality can be proved by Slater's theorem \cite{rockafellar1997convex}, when taking $B=\mathbb{I}$ and $A=\epsilon \mathbb{I}$ for $\epsilon\in (0,1)$, it is feasible for the dual.
	
	The remaining proof is similar to the proof of Corollary \ref{c1}, here we omit it.

	\subsection{Two classes of subchannels}
	Assume $\mathcal{H}_{AB}$ is a bipartite system with $dim(\mathcal{H}_A)=dim(\mathcal{H}_B)=m$. Let $\ket{\psi_m}=\frac{1}{\sqrt{m}}\sum_i\ket{ii}$ be the maximally entangled state(MES) of $\mathcal{H}_{AB}$. An important property of the MES is that it stays unchanged under the $\mathcal{T}(\cdot)$ operation, here
	\begin{align*}
		\mathcal{T}(\cdot)=\int_UdU (U\otimes  {U})^{\dagger}(\cdot)(U\otimes  {U}).
	\end{align*}
	Here $\mathcal{T}(\cdot)$ is local operations and shared randomness, hence, it can be realized by local operations and classical communication (LOCC).
	Next based on the Schur-Weyl theorem, $\mathcal{T}(X)$ can be written as follows,
	\begin{align*}
		\mathcal{T}(X)=\Psi_m\mathrm{tr}(X\Psi_m)+\tau_d\mathrm{tr}[X(\mathbb{I}-\Psi_m)],
	\end{align*}
	here $\Psi_m=\ket{\psi_m}\bra{\psi_m}$, $\tau_m=\frac{\mathbb{I}-\Psi_m}{m^2-1}$. Assume $\mathcal{N}$ is a subchannel, $i.$ $e.$, $\mathcal{N}$ is completely positive and trace-nonpreserving, then 
	\begin{align}
		\mathcal{N}\circ\mathcal{T}(X)=&\mathcal{N}(\Psi_m)\mathrm{ tr}(X\Psi_m )+\mathcal{N}(\tau_m)\mathrm{ tr}[X(\mathbb{I}-\Psi_m)],\label{sd1}\\
		\mathcal{T}\circ\mathcal{N}(X)=&\Psi_m\mathrm{ tr}\mathcal{N}(X)\Psi_m+\tau_m\mathrm{ tr}\mathcal{N}(X)(\mathbb{I}-\Psi_m),\label{sd2}
	\end{align}
	
	Next we consider two classes of 
	For (\ref{sd1}), as $\mathcal{N}(\cdot)$ is a subchannel, $\mathcal{N}(\Psi_m)$ and $\mathcal{N}(\tau_m)$ are substates, that is, 
	\begin{align*}
	\{\mathcal{N}\circ\mathcal{T}\}=\{\Lambda_{\gamma,\delta}|	\Lambda_{\gamma,\delta}(X)=\mathrm{ tr}(X\Psi_m)\cdot\gamma+\mathrm{tr}X(\mathbb{I}-\Psi_m)\cdot\delta, \gamma,\delta\in \overline{\mathcal{D}(\mathcal{H})}\}.
	\end{align*} For (\ref{sd2}), 
	\begin{align*}
		\Psi_m\mathrm{ tr}\mathcal{N}(X)\Psi_m+\tau_m\mathrm{ tr}\mathcal{N}(X)(\mathbb{I}-\Psi_m)=\Psi_m\mathrm{ tr}X\mathcal{N}^{\dagger}(\Psi_m)+\tau_m\mathrm{ tr}X\mathcal{N}^{\dagger}(\mathbb{I}-\Psi_m),\\
	\end{align*}
	 as $\mathcal{N}$ is a subchannel, then $\mathcal{N}^{\dagger}(\Psi_m),\mathcal{N}^{\dagger}(\mathbb{I}-\Psi_m)\le \mathcal{N}^{\dagger}(\mathbb{I})\le \mathbb{I}$, hence, 
	\begin{align*}
	\{\mathcal{T}\circ\mathcal{N}\}=	\{\Lambda_{M,N}|\Lambda_{M,N}(X)=\mathrm{ tr}MX\cdot \Psi_m+\mathrm{ tr}NX\cdot \tau_m,M+N\le I, M,N\ge 0\}.
	\end{align*}
Next we present properties of the subchannels $\Lambda_{\gamma,\delta}(X)$ and $\Lambda_{M,N}(X)$  needed here.
	\begin{lemma}\label{dm}
		Assume $\mathcal{H}_{AB}$ is a Hilbert space with $dim(\mathcal{H}_A)=dim(\mathcal{H}_B)=m,$ both $M$ and $N$ are semidefinite positive operators acting on $\mathcal{H}_{AB}$ with $M+N\le \mathbb{I}_{AB}$, let $\Lambda_{M,N}(X)=\mathrm{ tr}MX\cdot \Psi_m+\mathrm{ tr}NX\cdot \tau_m$, here $\tau_m=\frac{I-\Psi_m}{m^2-1}$, then for any $\varepsilon>0$, we have
		\begin{itemize}
			\item[(1).] $\Lambda_{M,N}(\cdot)\in \mathcal{NE}_{\varepsilon}$ if and only if $\sup_{X\in Sep}\frac{\mathrm{ tr}MX}{\mathrm{ tr}NX}\le \frac{2^{\epsilon}}{m-1}$.
			\item[(2).] $\Lambda_{M,N}(\cdot)\in \mathcal{DNE}_{\varepsilon}$ if and only if $\sup_{X\in Sep}\frac{\mathrm{ tr}MX}{\mathrm{ tr}NX}\le \frac{2^{\epsilon}}{m-1}$ and $N,\frac{1}{m}M+(1-\frac{1}{m})N\in \overline{Sep}.$
		\end{itemize}
	\end{lemma}
	
	\begin{proof}
		\begin{itemize}
			\item[(1).] As $D_{\Omega}(S,T)<\infty$ if and only $supp(S)=supp(T).$ For a substate $\gamma=p\Psi_m+q\tau_m$, it is separable if and only if $p\in [0,\frac{q}{m-1}]$. Besides, for any separable state $X$, $D_{\Omega}(\Lambda_{M,N}(X),\overline{Sep})\le 2^{\varepsilon}$ if and only if 
			$\frac{\max(\frac{\mathrm{tr}MX}{p},\frac{\mathrm{ tr}NX}{q})}{\min(\frac{\mathrm{ tr}MX}{p},\frac{\mathrm{ tr}NX}{q})}\le 2^{\varepsilon},$
			we finish the proof.
			\item[(2).] As  $\Lambda_{M,N}^{\dagger}(X)=\mathrm{tr}\Psi_mX\cdot M+\mathrm{ tr}\tau_mX\cdot N,$ and when $X$ is a separable state, $\mathrm{ tr}X\Psi_m\in (0,\frac{1}{m}).$ Hence, $\Lambda^{\dagger}_{M,N}(Sep)\subset\overline{Sep}$ if and only if $N\in \overline{Sep}$ and $\frac{1}{m}M+(1-\frac{1}{m})N\in \overline{Sep}.$ At last, based on the definition of $\mathcal{DNE}_{\epsilon}$, $\Lambda_{M,N}(\cdot)$ is $\mathcal{DNE}_{\varepsilon}$ if and only if $\Lambda_{M,N}\in \mathcal{NE}_{\varepsilon}$ and $\Lambda^{\dagger}_{M,N}(Sep)\subset cone(Sep)$, we finish the proof. 
		\end{itemize}
	\end{proof}
	\begin{lemma}\label{l4}
		Assume $\mathcal{H}_{AB}$ is a Hilbert space with $dim(\mathcal{H}_A)=dim(\mathcal{H}_B)=m,$ $\gamma$ and $\delta$ are two substates acting on $\mathcal{H}_{AB}.$ Let $\Lambda_{\gamma,\delta}(\rho)=\mathrm{ tr}(\rho\Psi_m)\gamma+\mathrm{tr}\rho(\mathbb{I}-\Psi_m)\delta,$ then for any $\varepsilon>0$, we have
		\begin{itemize}
			\item[(1).] $\Lambda_{\gamma,\delta}(\cdot)\in \mathcal{NE}_{\varepsilon
			}$ if and only if $	 \max{(D_{\Omega,\overline{Sep}}(\delta)),D_{\Omega,\overline{Sep}}((\frac{1}{m}\gamma+\frac{m-1}{m}\delta)))}<\varepsilon.$
			\item[(2).] $\Lambda_{\gamma,\delta}(\cdot)\in \mathcal{DNE}_{\varepsilon}$ if and only if $ \max{(D_{\Omega,\overline{Sep}}(\delta)),D_{\Omega,\overline{Sep}}((\frac{1}{m}\gamma+\frac{m-1}{m}\delta)))}<\varepsilon,$ and $\sup\limits_{\sigma\in Sep}\frac{\mathrm{tr}\sigma\gamma}{\mathrm{tr}\sigma\delta}\le m-1.$
			\item[(3).]Assume $\Lambda_{\gamma,\delta}\in \mathcal{NE}_{\varepsilon}$, then there exists $\delta^{'}$ such that 
			\begin{align*}
				\Lambda_{\gamma,\delta^{'}}(X)=\gamma\mathrm{tr}(X\Psi_{2m})+\delta^{'}\mathrm{tr}[X(\mathrm{I}-\Psi_m)]\in \mathcal{DNE}_{\epsilon}.
			\end{align*}
		\end{itemize}
	\end{lemma}
	\begin{proof}
		\begin{itemize}
			\item[(1).] As $\Lambda_{\gamma,\delta}\circ\mathcal{T}=\Lambda_{\gamma,\delta}$, and $\mathcal{T}$ is $\mathcal{NE}$, we only need to address the situation when the input state is an isotropic state. Moreover, as any separable isotropic state can be written as $p\frac{\mathbb{I}-\Psi_m}{m^2-1}+(1-p)(\frac{1}{m}\Psi_m+\frac{m-1}{m}\frac{\mathbb{I}-\Psi_m}{m^2-1})$, $p\in[0,1]$, then $\Lambda_{\gamma,\delta}\in NE_{\delta}$ if and only if
			\begin{align}
				\max{(D_{\Omega,\overline{Sep}}(\Lambda_{\gamma,\delta}(\tau_m)),D_{\Omega,\overline{Sep}}(\Lambda_{\gamma,\delta}(\frac{1}{m}\Psi_m+\frac{m-1}{m}\tau_m)))}<\varepsilon.
			\end{align}
			The above formula is due to that $D_{\Omega,Sep}(\cdot)$ is quasi-convex \cite{regula2022tight}.
			\item[(2).] As when $\Lambda$ is dually nonentangling, $\Lambda$ is nonentangling and $\Lambda^{\dagger}(Sep)\subset cone(Sep).$ Next $\Lambda_{\gamma,\delta}^{\dagger}(X)=\mathrm{tr}(X\gamma)\Psi+\mathrm{tr}(X\delta)(\mathbb{I}-\Psi)$, and $a\Psi_m+b(\mathbb{I}-\Psi_m)\in Sep$ if and only if $b\ge0$ and $a\in [0,b(m-1)],$ then we finish the proof.
			\item[(3).] As $\Lambda_{\gamma,\delta}\in \mathcal{NE}_{\epsilon},$ 
			\begin{align*}
				\max{(D_{\Omega,\overline{Sep}}(\delta)),D_{\Omega,\overline{Sep}}(\Lambda_{\gamma,\delta}(\frac{1}{m}\gamma+\frac{m-1}{m}\delta)))}	\le \epsilon,
			\end{align*} 
			let $\delta^{'}=\frac{1}{2m}\gamma+\frac{2m-1}{2m}\delta$, then 
			\begin{align*}
				D_{\Omega,\overline{Sep}}(\delta^{'})=&D_{\Omega,\overline{Sep}}(\frac{1}{2}(\frac{1}{m}\gamma+\frac{m-1}{m}\delta)+\frac{1}{2}\delta)\\
				\le& \max(D_{\Omega,\overline{Sep}}(\delta),D_{\Omega,\overline{Sep}}(\frac{1}{m}\gamma+\frac{m-1}{m}\delta))\\
				\le& \epsilon,
			\end{align*}
			next assume $\sigma$ is any separable state,
			\begin{align*}
				\frac{\mathrm{tr}\sigma\gamma}{\mathrm{ tr}\sigma\delta^{'}}=&\frac{\mathrm{tr}\sigma\gamma}{\mathrm{ tr}\sigma(\frac{1}{2m}\gamma+\frac{2m-1}{2m}\delta)}\\
				=&\frac{\frac{\mathrm{tr}\sigma\gamma}{\mathrm{ tr}\sigma\delta}}{\frac{\mathrm{ tr}\sigma\gamma}{2m\mathrm{\sigma\delta}}+\frac{2m-1}{2m}}\\
				\le &\frac{m-1}{\frac{m-1}{2m}+\frac{2m-1}{2m}}\le m-1,
			\end{align*}
			hence, based on (2), we finish the proof.
		\end{itemize}
	\end{proof}
	
	\subsection{Probabilistic entanglement distillation exponents under (approximately) $\mathcal{DNE}$ instruments}\label{aa1}
	
	In this section, we will analyse the error exponents of entanglement distillation  under the (approximately) $\mathcal{NE}(\mathcal{DNE})$ instruments. First we present the analytical formula of the error exponents of entanglement distillation, $E_{d,err,p}^{(m),\mathcal{F}_{\delta}}(\cdot),$  under $\mathcal{F}=\{\mathcal{NE},\mathcal{DNE}\}$ instruments.
	\begin{lemma}\label{lfd}
		Assume $\rho_{AB}$ is a bipartite state, then its probabilistic distillation exponent for the maximally entangled state $\ket{\psi_m}$ under the $\mathcal{F}_{\delta}$ instruments, $E_{d,err,p}^{(m),\mathcal{F}_{\delta}}(\rho_{AB})$, can be rewritten as
		\begin{align*}
			E_{d,err,p}^{(m),\mathcal{F}_{\delta}}(\rho_{AB})=&\sup\lim_{n\rightarrow\infty}-\frac{1}{n}\log\epsilon_n\\
			\textit{s. t.}\hspace{4mm}&  \frac{\mathrm{tr}M\rho_{AB}^{\otimes n}}{\mathrm{ tr}(M+N)\rho_{AB}^{\otimes n}}\ge 1-\epsilon_n,
			\mathcal{E}_i(X)=\mathrm{ tr}MX\cdot\Psi_m+\mathrm{ tr}NX\cdot\tau_m,\\
			&M+N\le \mathbb{I},M,N\ge 0,\\
			&\mathcal{E}_i\in \mathcal{E},\mathcal{E}\in \mathbb{O}_{\mathcal{F}_{\delta}},\mathcal{F}=\{\mathcal{NE},\mathcal{DNE}\}.
		\end{align*}
		Here $\tau_m=\frac{\mathbb{I}-\Psi_m}{m^2-1}$.
	\end{lemma}
	\begin{proof}
		As $\mathcal{T}(\cdot)=\int_U (U\otimes\overline{U})^{\dagger}(\cdot)(U\otimes\overline{U})\in \mathcal{F}$, when $\mathcal{E}_i(\cdot)\in \mathcal{E}$ and $\mathcal{E}\in\mathbb{O}_{\mathcal{F}_{\delta}},$ $\mathcal{T}\circ\mathcal{E}_i\in \mathcal{T}\circ\mathcal{E}$, $\mathcal{T}\circ\mathcal{E}\in\mathbb{O}_{\mathcal{F}_{\delta}}$.  then
		\begin{align*}
			F(\frac{\mathcal{E}_i(\rho_{AB}^{\otimes n})}{\mathrm{tr}(\mathcal{E}_i(\rho_{AB}^{\otimes n}))},\Psi_m)=&\bra{\psi_m}\frac{\mathcal{E}_i(\rho_{AB}^{\otimes n})}{\mathrm{ tr}\mathcal{E}_i(\rho_{AB}^{\otimes n})}\ket{\psi_m}\\
			=&\frac{\int_U\bra{\psi_m}(U\otimes \overline{U})^{\dagger}\mathcal{E}_i(\rho_{AB}^{\otimes n})(U\otimes\overline{U})\ket{\psi_m}dU}{\mathrm{ tr}\mathcal{E}_i(\rho_{AB}^{\otimes n})}\\
			=&\frac{\mathrm{ tr}\Psi_m(\mathrm{ tr}\mathcal{E}_i^{\dagger}(\Psi_m)\rho^{\otimes n}\cdot\Psi_m+\mathrm{ tr}\mathcal{E}_i^{\dagger}(\mathbb{I}-\Psi_m)\rho_{AB}^{\otimes n}\cdot\tau)}{\mathrm{ tr}\mathcal{E}_i(\rho_{AB}^{\otimes n})}\\
			=&\frac{\mathrm{ tr}M\rho_{AB}^{\otimes n}}{\mathrm{tr}(M+N)\rho_{AB}^{\otimes n}},
		\end{align*}
		Here $M=\mathcal{E}_i^{\dagger}(\Psi_m)$ and $N=\mathcal{E}_i^{\dagger}(\mathbb{I}-\Psi_m)$. As $\mathcal{E}_i$ is a subchannel, $M+N=\mathcal{E}_i(\mathbb{I})\le \mathbb{I}$.
		Hence, we finish the proof.
	\end{proof}

	\emph{Theorem \ref{t1}:} 
	Assume $\rho_{AB}^{\otimes n}$ is a bipartite state on $\mathcal{H}_{AB}^{\otimes n}$. Let $m\in \mathbb{N}$, $\delta\ge 0$, and the distillation is limited to $\mathcal{NE}_{\delta}$. Then the error exponent of the entanglement distillation is 
	\begin{align*}
		E_{d,err,p}^{(m),\mathcal{NE}_{\delta}}(\rho^{\otimes n}_{AB})=\frac{1}{n}\hat{\beta}_{\frac{2^{\delta}}{2^{\delta}+m-1},Sep}(\rho_{AB}^{\otimes n}).
	\end{align*}

Assume $\{m_n\}_n$ is a sequence of natural numbers with $\lim\limits_{n\rightarrow\infty}\frac{\log m_n}{n}=0,$ then
	\begin{align*}
		E_{d,err,p}^{\mathcal{NE}_{\delta}}(\rho_{AB})=\lim\limits_{n\rightarrow\infty}E_{d,err,p}^{(m_n),\mathcal{NE}_{\delta}}(\rho^{\otimes n}_{AB})=\hat{D}_{\Omega,Sep}^{reg }(\rho).
	\end{align*}
	\begin{proof}
		\emph{Meta-converse:}
		Assume $\{\mathcal{E}_i\}_{i=1}^k$ is a feasible quantum $\mathcal{NE}_{\delta}$ instruments such that
		\begin{align}
			1-\epsilon_n\le& F(\frac{\mathcal{E}_i(\rho^{\otimes n})}{\mathrm{tr}(\mathcal{E}_i(\rho^{\otimes n}))},\Psi_m)\nonumber\\
			=&\bra{\psi_m}\frac{\mathcal{E}_i(\rho^{\otimes n})}{\mathrm{tr}(\mathcal{E}_i(\rho^{\otimes n}))}\ket{\psi_m}\nonumber\\
			=&\frac{\mathrm{tr}[\rho_{AB}^{\otimes n}\mathcal{E}_i^{\dagger}(\Psi_m)]}{p_n},\label{tf0}
		\end{align}
		here $\mathcal{E}_i^{\dagger}(\cdot)$ satisfies $\mathrm{tr}(\mathcal{E}_i^{\dagger}(A)B)=\mathrm{tr}(A\mathcal{E}_i(B))$, $p_n=\mathrm{tr}(\mathcal{E}_i^{\dagger}(\mathbb{I})\rho_{AB}^{\otimes n})$.
		
		Let $M^{(n)}_2=\mathcal{E}_i^{\dagger}(\Psi_m), M^{(n)}_1=\mathcal{E}^{\dagger}_i(\mathbb{I}-\Psi_m)$, $M^{(n)}_0=\mathbb{I}-M^{(n)}_1-M_2^{(n)}$. As $\mathcal{E}_i$ is completely positive trace nonincreasing and $\mathbb{I}-\Psi_m\ge 0$, $M^{(n)}_1,M^{(n)}_2\ge 0$. As $\sum_{i=1}^k\mathcal{E}_i$ is trace preserving, then $\sum_i\mathcal{E}_i^{\dagger}(\mathbb{I})=\mathbb{I}$, 
		\begin{align*}
			M^{(n)}_0=&\mathbb{I}-M^{(n)}_1-M^{(n)}_2\\
			=&\sum_{\{1,2,\cdots,k\}-i}\mathcal{E}_l^{\dagger}(\mathbb{I})\ge 0,
		\end{align*} 
		Hence, $\{M^{(n)}_0,M^{(n)}_1,M^{(n)}_2\}$ is a POVM. Next based on 	(\ref{tf0}), we have
		\begin{align}
			\frac{\mathrm{tr}M_1^{(n)}\rho^{\otimes n}}{\mathrm{ tr}(M_1^{(n)}+M_2^{(n)})\rho^{\otimes n}}\le \epsilon_n.
		\end{align}
		
		Assume $\sigma_n$ is an arbitrary separable state in $\mathcal{H}_{AB}^{\otimes n}$, then
		\begin{align*}
			\frac{\mathrm{tr}M^{(n)}_2\sigma_n}{\mathrm{tr}(M^{(n)}_1+M^{(n)}_2)\sigma_n}=&\frac{\mathrm{tr}(\mathcal{E}_i(\sigma_n)\Psi_m)}{\mathrm{tr}\mathcal{E}_i(\sigma_n)}\\
			\le& \frac{2^{\delta}}{2^{\delta}+m-1},
		\end{align*}
		the last inequality is due to Lemma \ref{dm}. Thus, based on the definition of postselected hypothesis testing, we have
		\begin{align*}
			\hat{\beta}_{\frac{2^{\delta}}{2^{\delta}+m-1}, Sep}(\rho_{AB}^{\otimes n})\ge& \log\frac{\mathrm{tr}(M^{(n)}_1+M^{(n)}_2)\rho^{\otimes n}}{\mathrm{tr}(M^{(n)}_1\rho^{\otimes n})}\\
			\ge&-\log\epsilon_n.
		\end{align*}
		Then multiplying two sides $\frac{1}{n}$, and when taking the supremum over all $\{\mathcal{E}_i\}_{i=1}^k\in NE,$ we have
		\begin{align}
			E^{(m),\mathcal{NE}_{\delta}}_{d,err,p}(\rho_{AB} )\le&	\frac{1}{n}\min_{\sigma_n\in Sep_{A_n:B_n}} \hat{\beta}_{\frac{2^{\delta}}{2^{\delta}+m-1}}(\rho_{AB}^{\otimes n},\sigma_n)\nonumber\\
			=&\frac{1}{n} \hat{\beta}_{\frac{2^{\delta}}{2^{\delta}+m-1},Sep}(\rho_{AB}^{\otimes n})\label{a1}
		\end{align}

		$Achievability:$ Assume $\{M^{(n)}_1,M^{(n)}_2,M^{(n)}_0|M^{(n)}_i\ge 0,i=1,2,0\}$ is a POVM that achieves the optimal type II error $\hat{\beta}_{\frac{2^{\delta}}{2^{\delta}+m-1},Sep}(\rho^{\otimes n})$ with $\frac{\mathrm{tr}M^{(n)}_2\sigma_n}{\mathrm{ tr}M_1^{(n)}\sigma_n}\le \frac{2^{\delta}}{m-1}$ for any $\sigma_n\in Sep(A_n:B_n)$. Next we construct the following subchannel $\mathcal{E}_1(\cdot)$, $\mathcal{E}_2(\cdot)$,
		\begin{align*}
			\mathcal{E}_1(\cdot)=\mathrm{tr}(M^{(n)}_2(\cdot))\Psi_m+\mathrm{tr}M^{(n)}_1(\cdot)\frac{\mathbb{I}-\Psi_m}{m^2-1},\\
			\mathcal{E}_2(\cdot)=\mathrm{tr}M_0^{(n)}(\cdot)\frac{\mathbb{I}-\Psi_m}{m^2-1}.
		\end{align*}
		For any $\sigma_n\in Sep_{A_n:B_n}$, as $\mathbb{I}-\Psi_m$ is separable, $\mathcal{E}_2$ is $\mathcal{NE}$. Based on Lemma \ref{dm}, and  $\frac{\mathrm{tr}M^{(n)}_2\sigma_n}{\mathrm{ tr}M_1^{(n)}\sigma_n}\le \frac{2^{\delta}}{m-1}$, $\mathcal{E}_1(\cdot)$ is in $\mathcal{NE}_{\delta}.$ As $\mathrm{tr}(\mathcal{E}_1(\cdot)+\mathcal{E}_2(\cdot))=\mathrm{tr}(\cdot)$, and $\mathcal{E}_1$ and $\mathcal{E}_2$ is completely positive, $\{\mathcal{E}_1,\mathcal{E}_2\}$ is an $\mathcal{NE}$ instruments. Then we have
		\begin{align*}
			E_{d,err,p}^{(m),\mathcal{NE}_{\delta}}(\rho^{\otimes n}_{AB} )\ge& -\frac{1}{n}\log\frac{\mathrm{tr}M_1^{(n)}\rho^{\otimes n}}{\mathrm{tr}(M_1^{(n)}+M_2^{(n)})\rho^{\otimes n}},
		\end{align*}
		when taking over all measurements, we have
		\begin{align}
			E_{d,err,p}^{(m),\mathcal{NE}_{\delta}}(\rho_{AB}^{\otimes n} )\ge& \frac{1}{n}\hat{\beta}_{\frac{2^{\delta}}{2^{\delta}+m-1},Sep}(\rho^{\otimes n}).\label{a2}
		\end{align}
		
		Based on (\ref{a1}) and (\ref{a2}), we have
		\begin{align}
			E_{d,err,p}^{(m),\mathcal{NE}_{\delta}}(\rho^{\otimes n} )= \frac{1}{n}\hat{\beta}_{\frac{2^{\delta}}{2^{\delta}+m-1},Sep}(\rho^{\otimes n})\label{a3}
		\end{align}

Assume $\{m_n\}_{n}$ is sequence with $\lim\limits_{n\rightarrow\infty}\frac{\log m_n}{n}=0,$ \begin{align*}
	\lim\limits_{n\rightarrow\infty}\frac{1}{n}\log\frac{1}{1-\epsilon_n}=\lim\limits_{n\rightarrow\infty}\frac{1}{n}(\log(2^{\delta}+m_n-1)-\log(m_n-1))=0,\\
	\lim\limits_{n\rightarrow\infty}\frac{1}{n}\log\frac{\epsilon_n}{(1-\epsilon_n)}=\lim\limits_{n\rightarrow\infty}\frac{1}{n}(\delta-\log(m_n-1))=0.
\end{align*}
combining Corollary \ref{c1} and (\ref{a3}), we have
		\begin{align*}
			E^{\mathcal{NE}_{\delta}}_{d,err,p}(\rho_{AB})=\lim\limits_{n\rightarrow\infty}\frac{1}{n}\hat{\beta}_{\frac{2^{\delta}}{2^{\delta}+m_n-1},Sep}(\rho^{\otimes n})=&\hat{D}_{\Omega,Sep}^{reg}(\rho).
		\end{align*}
		Hence, we finish the proof.
	\end{proof}
	
	\emph{Theorem \ref{th2}} 	Assume $\rho_{AB}^{\otimes n}$ is a bipartite state on $\mathcal{H}_{AB}^{\otimes n}$. Let $\delta\ge 0$ and $\{m_n\}_n$ is a sequence of natural numbers. Then the error exponent of the entanglement distillation is 
	\begin{align*}
	\hat{\beta}^{\mathbb{SEP}}_{\frac{1}{m_n+1}+\frac{2^{\delta}m_n}{(m_n+1)(2^{\delta}+m_n-1)},Sep}(\rho_{AB}^{\otimes n})- \log\frac{m_n+1}{m_n}\ge	nE_{d,err,p}^{(m_n),\mathcal{DNE}_{\delta}}(\rho_{AB}^{\otimes n})\ge& \hat{\beta}_{\frac{2^{\delta}}{2^{\delta}+m_n-1},Sep}^{\mathbb{SEP}}(\rho_{AB}^{\otimes n}).
		\end{align*}
	Furthermore, when $\{m_n\}_n$ is sequence with $\lim\limits_{n\rightarrow\infty}\frac{\log m_n}{n}=0,$ let $n\rightarrow \infty,$
	\begin{align*}
		E^{\mathcal{DNE}_{\delta}}_{d,err,p}(\rho_{AB})=\hat{D}_{\Omega,Sep}^{reg,\mathbb{SEP}}(\rho_{AB}).
	\end{align*}
	\begin{proof}
	Assume $\{M^{(n)},N^{(n)},I-M^{(n)}-N^{(n)}\}$ is the optimal for $\rho_{AB}^{\otimes n}$ for $\hat{\beta}^{\mathbb{SEP}}_{\frac{2^{\delta}}{2^{\delta}+m_n-1},Sep}(\rho^{\otimes n})$, then $M^{(n)},N^{(n)},I-M^{(n)}-N^{(n)}\in cone(Sep)$, and $\sup_{\sigma\in Sep}\frac{\mathrm{ tr}M^{(n)}\sigma}{\mathrm{ tr}N^{(n)}\sigma}\le \frac{2^{\delta}}{m_n-1}.$ then $N^{(n)},M^{(n)}+(m-1)N^{(n)}\in cone(Sep)$. Based on Lemma \ref{lfd}, for $E_{d,err,p}^{(m_n),\mathcal{F}_{\delta}}(\rho_{AB}^{\otimes n})$, we could always assume $\mathcal{E}_i(X)=\mathrm{ tr}MX\cdot\Psi_{m_n}+\mathrm{ tr}NX\cdot\tau_{m_n}.$ Based on Lemma \ref{dm}, when choosing $M=N^{(n)},$ and $N=M^{(n)}$, $\mathcal{E}_1(\cdot)\in \mathcal{DNE}_{\epsilon}.$ Then 
		\begin{align}
			E_{d,err,p}^{(m_n),\mathcal{DNE}_{\delta}}(\rho_{AB}^{\otimes n})\ge -\frac{1}{n}\log\frac{\mathrm{ tr}N^{(n)}\rho^{\otimes n}}{\mathrm{ tr}(N^{(n)}+M^{(n)})\rho^{\otimes n}}=\frac{1}{n}\hat{\beta}_{\frac{2^{\delta}}{2^{\delta}+m_n-1},Sep}^{\mathbb{SEP}}(\rho_{AB}^{\otimes n}).\label{th2f1}
		\end{align}

			Next assume $\mathcal{E}=\{(\mathcal{E}_1,\mathcal{E}_2)|\mathcal{E}_1,\mathcal{E}_2\in \mathcal{DNE}_{\delta},\mathrm{tr}[\mathcal{E}_1(\cdot)+\mathcal{E}_2(\cdot)]=\mathrm{ tr}(\cdot)\}$ is a feasible one-shot distillation protocol such that 
		\begin{align*}
			F(\frac{\mathcal{E}_1(\rho^{\otimes n})}{\mathrm{ tr}\mathcal{E}_1(\rho^{\otimes n})},\Psi_{m_n})=\mathrm{ tr}\frac{\mathcal{E}_1(\rho^{\otimes n})}{\mathrm{ tr}\mathcal{E}_1(\rho^{\otimes n})}\Psi_{m_n}\ge 1-\epsilon.
		\end{align*}
		Here we always assume $\mathcal{E}_1(X)=\mathrm{ tr}MX\cdot\Psi_{m_n}+\mathrm{ tr}NX\cdot\tau_{m_n}$, $M+N\le \mathbb{I},$ $M,N\ge 0.$ As $\mathcal{E}_1\in\mathcal{DNE}_{\delta}$, $\sup_{X\in Sep}\frac{\mathrm{ tr}MX}{\mathrm{ tr}NX}\le \frac{2^{\delta}}{m_n-1}$ and $N,M+(m_n-1)N\in cone(Sep)$. Then let $M_2=\mathcal{E}_1^{\dagger}(\frac{1}{m_n+1}(\mathbb{I}+m_n\Psi_{m_n})),M_1=\mathcal{E}_1^{\dagger}(\frac{m_n}{m_n+1}(\mathbb{I}-\Psi_{m_n})),$ for any separable state $\sigma$,
		\begin{align*}
			\frac{\mathrm{ tr}M_2\sigma}{\mathrm{ tr}(M_1+M_2)\sigma}=&\frac{\mathrm{ tr}(\frac{1}{m_n+1}(\mathbb{I}+m_n\Psi_{m_n}))\mathcal{E}_1(\sigma)}{\mathrm{tr}(\mathcal{E}_1(\sigma))}\\
			\le&\frac{1}{m_n+1}+\frac{2^{\delta}m_n}{(m_n+1)(2^{\delta}+m_n-1)},
		\end{align*}
		here the last inequality is due to the definition of $\mathcal{DNE}_{\delta}$,
		then 
		\begin{align}
		\frac{1}{n}	\hat{\beta}^{\mathbb{SEP}}_{\frac{1}{m_n+1}+\frac{2^{\delta}m_n}{(m_n+1)(2^{\delta}+m_n-1)},Sep}(\rho_{AB}^{\otimes n})\ge& \frac{1}{n}\log \frac{\mathrm{ tr}(M_1+M_2)\rho_{AB}^{\otimes n}}{\mathrm{ tr}M_1\rho_{AB}^{\otimes n}}\nonumber\\
			=& \frac{1}{n}\log\frac{1}{\mathrm{ tr}\frac{m_n}{m_n+1}(\mathbb{I}-\Psi_{m_n})\frac{\mathcal{E}_1(\rho_{AB}^{\otimes n})}{\mathrm{tr}\mathcal{E}_1(\rho^{\otimes n})}}\nonumber\\
			\ge& E_{d,err,p}^{(m_n),\mathcal{DNE}_{\delta}}(\rho^{\otimes n}_{AB})+\frac{1}{n}\log\frac{m_n+1}{m_n} \label{th2f2}
		\end{align}
		Here the first inequality is due to the definition of $\hat{\beta}^{\mathbb{SEP}}_{\epsilon,Sep}(\rho)$.
		
		At last, based on (\ref{th2f1}) and (\ref{th2f2}), we have
		\begin{align}
			\lim\limits_{n\rightarrow\infty}\frac{1}{n}\hat{\beta}_{\frac{2^{\delta}}{2^{\delta}+m_n-1},Sep}^{\mathbb{SEP}}(\rho_{AB}^{\otimes n})\le\lim\limits_{n\rightarrow\infty}\frac{1}{n} E_{d,err,p}^{(m_n),\mathcal{DNE}_{\delta}}(\rho_{AB}^{\otimes n})\label{th2f3},\\
			\lim_{n\rightarrow\infty}\frac{1}{n}	\hat{\beta}^{\mathbb{SEP}}_{\frac{1}{m_n+1}+\frac{2^{\delta}m_n}{(m_n+1)(2^{\delta}+m_n-1)},Sep}(\rho^{\otimes n})\ge\lim\limits_{n\rightarrow\infty}\frac{1}{n} E_{d,err,p}^{(m_n),\mathcal{DNE}_{\delta}}(\rho_{AB}^{\otimes n})\label{th2f4}.
		\end{align}
		Hence, combining $(\ref{th2f3})$, (\ref{th2f4}) and Corollary \ref{c2}, we have
		\begin{align*}
			E_{d,err,p}^{\mathcal{DNE}_{\delta}}(\rho_{AB})=\hat{D}_{\Omega,Sep}^{reg,\mathbb{SEP}}(\rho).
		\end{align*}

	\end{proof}

	\emph{Example \ref{e1}:}Assume $\mathcal{H}_{AB}$ is a bipartite system with $dim(\mathcal{H}_A)=dim(\mathcal{H}_B)=d$, and $\rho_{AB}$ is the Werner state,
	\begin{align*}
		\rho_{p}=p\cdot\frac{2P_s}{d(d+1)}+(1-p)\cdot\frac{2P_{as}}{d(d-1)},
	\end{align*}
	here $P_s=\frac{I+F}{2}$, $P_{as}=\frac{I-F}{2}$, $F$ is the swap operator, $F=\sum_{ij}\ket{ij}\bra{ji}$. Then 
\begin{itemize}
	\item[(i).]
For each $n\in\mathbb{N}$,
	\begin{equation*}
		\frac{1}{n}D_{\Omega,Sep}(\rho_p^{\otimes n})= \begin{cases} 
			\log\frac{1-p}{p}\hspace{3mm}p<\frac{1}{2} \\
			0\hspace{12mm}p\ge\frac{1}{2}
		\end{cases}  .
	\end{equation*}
	Hence, 
	\begin{align*}
		\hat{D}_{\Omega,Sep}^{reg}(\rho_p)=\begin{cases} 
			\log\frac{1-p}{p}\hspace{3mm}p<\frac{1}{2} \\
			0\hspace{12mm}p\ge\frac{1}{2}
		\end{cases}  .
	\end{align*}
\item [	(ii).] 	\begin{align*}
	\hat{D}_{Sep}^{\mathbb{SEP}}(\rho^{\otimes n}_p)=\begin{cases} 
		\log\frac{d+1-2p}{2dp}\hspace{3mm}0<p<\frac{1}{2} \\
		0\hspace{20mm}p\ge\frac{1}{2}
	\end{cases}  .
\end{align*}
\end{itemize} 
	\begin{proof}\begin{itemize}
			\item[(i).] 
		Due to the faithfulness of $D_{\Omega}(\cdot)$ in Lemma \ref{l1}, and when $p\ge \frac{1}{2}$, $\rho_{AB}$ is separable \cite{reduction1999}, we only need to consider the case when $p<\frac{1}{2}$.
		
		As $p_{\frac{1}{2}}$ is separable,
		\begin{align*}
			\frac{1}{n}D_{\Omega,Sep}(\rho_{p}^{\otimes n})\le& \frac{1}{n}D_{\Omega}(\rho_{p}^{\otimes n},\rho_{\frac{1}{2}}^{\otimes n})\\
			=&D_{\Omega}(\rho_p,\rho_{\frac{1}{2}}),
		\end{align*} 
		As $P_s$ and $P_{as}$ are two mutually orthogonal projectors, then 
		\begin{align*}
			D_{\Omega}(\rho_p,\rho_{\frac{1}{2}})=&\hspace{2mm}\log\inf\hspace{4mm} \mu\\
			\textit{s. t.}\hspace{3mm}&(1,1)\le(2p\lambda,2(1-p)\lambda)\le(\mu,\mu),\\
			&\lambda,\mu\ge 0.
		\end{align*}
		From computation, we have 
		\begin{align}
			D_{\Omega,Sep}(\rho_p)\le	D_{\Omega}(\rho_p,\rho_{\frac{1}{2}})=\log\frac{1-p}{p}. \label{lf1}\\
			\frac{1}{n}D_{\Omega,Sep}(\rho_p^{\otimes n})\le		\frac{1}{n}D_{\Omega}(\rho^{\otimes n}_p,\rho^{\otimes n}_{\frac{1}{2}})=\log\frac{1-p}{p}.\label{lf2}
		\end{align}Here $n$ is an arbitrary natural number. Next we show the other direction. 
		
		Next we show the dual problem of $\Omega_{Sep}(\cdot)$,
		\begin{align}
			\Omega_{Sep}(\rho)=\hspace{2mm}&\hspace{2mm}\sup\mathrm{tr}(A\rho)\label{dsep}\\
			\textit{s. t.}\hspace{4mm}&\mathrm{tr}B\rho=1,\nonumber\\
			&\mathrm{tr}(B-A)\sigma\ge 0\hspace{3mm}\forall \sigma\in Sep,\nonumber\\
			&A,B\ge 0.\nonumber
		\end{align} 
		Next let
		\begin{align*}
			A=\frac{1}{p}P_{as},\hspace{3mm}B=\frac{1}{p}P_{s}.
		\end{align*}
		Due to computation, $\mathrm{tr}A\rho=\frac{1-p}{p},$ $\mathrm{tr}B\rho=1$. Next we show the last condition, when $\sigma$ is any separable state,
		\begin{align*}
			\mathrm{tr}(B-A)\sigma=&\frac{1}{p}\mathrm{tr}\frac{I+F-I+F}{2}\sigma\\\
			=&		\frac{1}{p}\mathrm{ tr}F\sigma\ge 0.
		\end{align*}
		Hence, $A$ and $B$ are feasible for the dual program of $\Omega_{Sep}(\cdot)$, then
		\begin{align}
			D_{\Omega,{Sep}}(\rho)\ge \log\frac{1-p}{p}. \label{rf1}
		\end{align}
		Combining $(\ref{lf1})$ and (\ref{rf2}), we have $$D_{\Omega,Sep}(\rho_p)=\log\frac{1-p}{p}.$$ For the state $\rho_p^{\otimes n}$, let 
		\begin{align*}
			A^{(n)}=A^{\otimes n},B^{(n)}=B^{\otimes n},
		\end{align*}
		Due to the computation, $\mathrm{tr}B^{(n)}\rho_p^{\otimes n}=1,$ for any separable state $\sigma_n\in Sep_{A_n:B_n}$,
		\begin{align*}
			&\mathrm{tr}(B^{(n)}-A^{(n)})\sigma_n\\
			&=\mathrm{tr}(\Pi_1+\Pi_3+\cdots+\Pi_{2\lceil\frac{n}{2}\rceil-1})\sigma_n,
		\end{align*}
		here $\Pi_m$ is a sum of all product opertors with $m$ $F$ and $n-m$ $I$. As each $\Pi_m$ satisfies $\mathrm{tr}\Pi_m\sigma\ge 0$, the above formula is nonnegative. Hence $A^{(n)}$ and $B^{(n)}$ are the feasible for $\rho_p^{\otimes n}$ in terms of $(\ref{dsep})$ for $\Omega_{Sep}$, then
		\begin{align}
			\frac{1}{n}	D_{\Omega,Sep}(\rho^{\otimes n})\ge \log\frac{1-p}{p},\label{rf2}
		\end{align}
		combining (\ref{lf2}) and (\ref{rf2}), we have $$\frac{1}{n}D_{\Omega,Sep}(\rho^{\otimes n}_p)=\log\frac{1-p}{p}.$$
		
		\item[(ii).]  Based on Corollary \ref{c2}, 
		\begin{align}
			\Omega_{Sep}^{\mathbb{SEP}}(\rho_p)=&\inf \gamma\label{dsep}\\
			\textit{s. t.}\hspace{3mm}&\rho_p\preceq_{\mathbb{SEP}^{*}}t{\sigma}\preceq_{\mathbb{SEP}^{*}}\gamma\rho_p,\nonumber\\
		&	{\sigma}\in Sep, t>0.\nonumber
		\end{align}
		
		As $\rho_p$ satisfies $(U\otimes{U})^{\dagger} \rho_p(U\otimes{U})=\rho_p$, and $\int_UdU(U\otimes{U})^{\dagger} {\sigma}(U\otimes {U})=a \frac{2P_s}{d(d+1)}+(1-a)\frac{2P_{as}}{d(d-1)},$ as ${\sigma}\ge 0$, $a\ge \frac{ 1}{2}.$ As for any Werner state $W=m\frac{2P_s}{d(d+1)}+n\frac{2P_{as}}{d(d-1)}$, 
		\begin{align*}
			W\succeq_{\mathbb{SEP}^{*}}0\Longleftrightarrow \mathrm{ tr}W(\ket{xy}\bra{xy})=\frac{m(1+|\bra{x}y\rangle|^2)}{d(d+1)}+\frac{n(1-|\bra{x}y\rangle|^2}{d(d-1)})\ge0\Longleftrightarrow m\ge 0, (d-1)m+(d+1)n\ge 0,
		\end{align*}

		 Hence, (\ref{dsep}) turns into the following,
		\begin{align*}
			\Omega_{Sep}^{\mathbb{SEP}}(\rho_p)=&\inf\gamma\\
			\textit{s. t.}\hspace{3mm}& 	ta-p\ge 0, (d-1)(ta-p)+(d+1)(t(1-a)-1+p)\ge 0,\\
		&	\gamma p-ta\ge 0, (d-1)(\gamma p-ta)+(d+1)(\gamma-\gamma p-t(1-a))\ge 0, t>0.
		\end{align*}
		As when $p\ge\frac{1}{2},$ $\rho_p$ is separable, we only consider $0<p<\frac{1}{2}.$ Due to the computation, the optimal $\sigma\in Sep$ in \ref{dsep} is $\rho_{\frac{1}{2}},$ we have
		\begin{align*}
				\hat{D}_{Sep}^{\mathbb{SEP}}(\rho_p)=\begin{cases} 
				\log\frac{d+1-2p}{2dp}\hspace{3mm}0<p<\frac{1}{2} \\
				0\hspace{20mm}p\ge\frac{1}{2}
			\end{cases}  .
		\end{align*}
	 
		\end{itemize}
		
	\end{proof}

\end{document}